\documentclass[12pt,a4paper]{amsart}
\usepackage{amsfonts}
\usepackage{amsthm}
\usepackage{amsmath}
\usepackage{amssymb}
\usepackage{amscd}
\usepackage{marvosym}
\usepackage{bm}
\usepackage{bbm}
\usepackage{braket}
\usepackage[latin2]{inputenc}
\usepackage{t1enc}
\usepackage[mathscr]{eucal}
\usepackage{enumitem}
\usepackage{indentfirst}
\usepackage{float}
\usepackage{mathtools}
\usepackage{faktor}
\usepackage{tikz}
\usepackage{caption}
\usepackage{xcolor}
\usepackage{makecell}
\usepackage[hidelinks]{hyperref}

\numberwithin{equation}{section}
\usepackage[margin=2.9cm]{geometry}

\newcommand{\Larrow}[1]{
\parbox{#1}{\tikz{\draw[->](0,0)--(#1,0);}}
}

\theoremstyle{plain}
\newtheorem{thm}{Theorem}[section]
\newtheorem{lem}[thm]{Lemma}
\newtheorem{cor}[thm]{Corollary}
\newtheorem{prop}[thm]{Proposition}

\theoremstyle{definition}
\newtheorem{defn}[thm]{Definition}

\newtheorem{remk}[thm]{Remark}

\DeclareMathOperator{\Hom}{Hom}
\DeclareMathOperator{\Ob}{Ob}

\DeclareMathOperator{\Aut}{Aut}
\DeclareMathOperator{\FdHilb}{FdHilb}

\title{Fusion Structure from Exchange Symmetry in (2+1)-Dimensions}
\author[Sachin J. Valera]{Sachin J. Valera$^{\dagger}$}
\thanks{$^{\dagger}$\textit{Selmer Center, Department of Informatics, University of Bergen, Norway}}

\begin{document}
\begin{abstract}Until recently, a careful derivation of the fusion structure of anyons from some underlying physical principles has been lacking. In [Shi et al., Ann. Phys., 418 (2020)], the authors achieved this goal by starting from a conjectured form of entanglement
area law for 2D gapped systems. In this work, we instead start with the principle of exchange symmetry, and determine the minimal prescription of additional postulates needed to make contact with unitary ribbon fusion categories as the appropriate {algebraic} framework for modelling anyons. Assuming that 2D quasiparticles are spatially localised, we build a functor from the coloured braid groupoid to the category of finite-dimensional Hilbert spaces. Using this functor, we construct a precise notion of exchange symmetry, allowing us to recover the core fusion properties of anyons. In particular, given a system of $n$ quasiparticles, we show that the action of a certain $n$-braid $\beta_{n}$ uniquely specifies its superselection sectors. We then provide an overview of the braiding and fusion structure of anyons in the usual setting of braided $6j$ fusion systems. By positing the duality axiom of [A. Kitaev, Ann. Phys., 321(1) (2006)] and {assuming} that there are finitely many distinct topological charges, we arrive at the framework of ribbon categories. 
\end{abstract}
\maketitle

\vspace{-4.5mm}
\section{Introduction} 
\label{introduction}
The study and classification of topological phases of matter is a pervasive theme of contemporary physics. Quasiparticles with exotic exchange statistics (called ``anyons'') are a hallmark of two-dimensional topological phases. The experimental realisation and control of anyons is a much sought-after goal, owing especially to proposed schemes for the robust processing of quantum information \cite{kitqc,univmf,tqc}.\\ 

The algebraic theory of anyons (of which various detailed accounts may be found \cite{kitaev,bonderson,simonox,preskill,wangbook}) is considered mature \cite{mtqc,beyond}. It is well-understood that the statistical properties of anyons arise due to the distinguished topology of exchange trajectories in two dimensions. In a given theory, anyons are distinguished by their ``topological charges'' which characterise their mutual statistics. However, it is further expected that these charges possess a fusion structure wherein the `combination' (or \textit{fusion}) of two anyons effectively results in a single anyon that may possibly exist in a superposition of topological charges. In some expositions, fusion is motivated using flux-charge composite toy models. Fusion structure is also readily apparent in 2D spin-lattice models such as the toric code. However, a careful treatment of the emergence of this fusion structure in a general setting is lacking. We therefore seek to provide a ground-up construction of the braiding and fusion structure of anyons. \\ 

\textit{Quantum symmetries} is an umbrella term for the algebraic structures that are used to describe topological quantum matter. Ribbon fusion categories provide the mathematical framework for studying the statistical behaviour of anyons. Often, anyons are introduced through a discussion of identical particles: the same arguments that lead us to conclude that there are only bosons and fermions in three or more spatial dimensions, instead indicate the possibility of fractional statistics in two dimensions. There is an unfortunate gulf between the language of identical particles and that of ribbon categories. Our objective is to clarify the connection between quantum symmetries and the elementary, yet profound principle of \textit{exchange symmetry} in quantum mechanics. Superselection sectors play a key role in our exposition. \\

A series of `assumptions' or postulates \textbf{A1}-\textbf{A3} are given throughout the text. They are proposed as the minimal prescription needed to recover ribbon fusion categories (as an algebraic model for anyons) from exchange symmetry in $(2+1)$-dimensions. Here, \textbf{A2}-\textbf{A3} are presented in terms slightly more simplified than in the main text. The ``associativity condition'' in \textbf{A3} refers to (\ref{brazlioti}). \\

\hspace{-6mm}\fbox{
\parbox{\textwidth}{
\noindent\textbf{A1.} Two-dimensional quasiparticles are \textit{spatially localised} phenomena. \\ 

\noindent\textbf{A2.} (i) The Hilbert space of finitely many quasiparticles is finite-dimensional. \\ \phantom{a}\quad \hspace{1.5mm} (ii) A theory of anyons has finitely many distinct topological charges.\\

\noindent\textbf{A3.} For any topological charge $q$, there exists a dual charge $\bar{q}$ such that a \\ \phantom{a}\quad \hspace{1.5mm} certain associativity condition is satisfied with respect to their fusion.} \\
}\\[3mm]

\noindent The \textit{localisation condition} \textbf{A1} is a relevant physical consideration. Less satisfyingly, \textit{finiteness assumption} \textbf{A2} appears to be prescribed for mathematical convenience. Some physical motivation is provided for \textit{Kitaev's duality axiom} \textbf{A3} in \cite{kitaev}. The main results of this paper are presented in Sections \ref{exches2} and \ref{superbraidsec}, where we show that \textbf{A1} and \textbf{A2}(i) are sufficient to recover the core braiding and fusion structure of 2D quasiparticles. In Section \ref{anyonsec}, we outline how \textbf{A2}(ii) and \textbf{A3} are required to make contact with ribbon fusion categories as algebraic models for anyons.

\subsection{Relation to existing work}\label{relshisec} In attempting to derive fusion structure from some underlying physical principles, our work is similar in spirit to \cite{shikokim} where the authors show that such structure may be recovered from the entanglement area law
\begin{equation}S(A)=\alpha l - \gamma \label{EAL}\end{equation}
where $S(A)$ is the von Neumann entropy of a simply-connected region $A$, $l$ is the perimeter of $A$ and $\gamma$ is a constant correction term (which the authors also show to be equal to $\ln\mathcal{D}$, where $\mathcal{D}$ is the total quantum dimension of the anyon theory).
\vspace{-2mm}
\def\arraystretch{1.4}
{\small
\begin{table}[H]
\centering
\begin{tabular}{|c|c|c|} 
\hline
  & \textbf{Our approach} & \textbf{Approach in \cite{shikokim}} \\ \hline 
 \textit{Physical principle} & Exchange symmetry & Entanglement area law \\ \hline
 \textit{Construction} & \makecell{Local representations of \\ coloured braid groupoid} & Information convex sets \\ \hline
\end{tabular}
\vspace{1mm}
\label{comparitable}
\end{table}
}
\def\arraystretch{1}
\vspace{-4mm}
\noindent While the construction in \cite{shikokim} may be more fundamental, the narrative of exchange symmetry might be more familiar to the majority of readers. $F$ and $R$ symbols can be recovered from our construction, and we are able to arrive at the usual formalism (of unitary ribbon fusion categories) for modelling theories of anyons. Ultimately, the two approaches will offer different insights and will appeal to different audiences. However, we suggest that they might be viewed as complementing one another. By assuming (\ref{EAL}) it follows that \textbf{A2}(i) implies \textbf{A2}(ii) \cite[Theorem 4.1]{shikokim}, and that for any topological charge $q$ there exists a \textit{unique} dual charge $\bar{q}$ such that they will fuse to the vacuum in a \textit{unique} way \cite[Proposition 4.9]{shikokim}. Combining the two approaches, we arrive at an alternative to \textbf{A1}-\textbf{A3}:\footnote{The authors of \cite{shikokim} advocate for using two local entropic constraints \cite[A0-A1]{shikokim} in lieu of (\ref{EAL}). }\\

\hspace{-6mm}\fbox{
\parbox{\textwidth}{
\noindent\textbf{P1.} Two-dimensional quasiparticles are \textit{spatially localised} phenomena. \\ 

\noindent\textbf{P2.} The Hilbert space of finitely many quasiparticles is finite-dimensional. \\

\noindent\textbf{P3.} The system of quasiparticles satisfies entanglement area law (\ref{EAL}).
}
}

\subsection{Outline of paper} In Section \ref{prelimsec}, we recap the notion of superselection rules and identical particles. This is followed by a discussion of the difference between particle exchanges in two and three spatial dimensions.  In Section \ref{exches3}, we formulate exchange symmetry via the action of the motion group of a many-particle system, and relate this to the boson-fermion superselection rule for fundamental particles.\\ 

In Section \ref{exches2}, we consider the action of braiding on a system of 2D quasiparticles. The localisation condition \textbf{A1} means that this action is generally not given by a representation of the braid group; instead, it is given by a local representation of the ``coloured'' braid groupoid. This action is described in Section \ref{quab}, and we discuss its interpretation as a functor in Appendix \ref{cbgactappx}. The heart of our construction is presented in Section \ref{exches2dsec}, where we adapt the definition of exchange symmetry from Section \ref{exches3} to formulate an appropriate commutator via the braiding action. This gives rise to a notion of exchange symmetry on all subsystems of quasiparticles. In Section \ref{supsec2dsec}, we see how the associated superselection sectors of subsystems fit together to describe the Hilbert space of the whole system. \\

In Section \ref{superbraidsec}, we present our main results. We show that the superselection sectors of an $n$-quasiparticle system correspond to the eigenspaces under the action of an $n$-braid $\beta_{n}$ which we call the \textit{superselection braid} (Theorem \ref{main1}). We recover the core fusion structure amongst these superselection sectors by showing that they exhibit the same statistical behaviour as quasiparticles, allowing us to identify them as such (Theorem \ref{mainthm2}). The associativity and commutativity of fusion is deduced in Corollary \ref{fucoma}. We prove several braid identities pertaining to $\beta_{n}$ and see that this braid encodes the structure of all fusion trees for an $n$-quasiparticle fusion space (Theorem \ref{ssbrecur}). We finally show that $\beta_{n}$ is the unique braid (up to orientation) whose action specifies the superselection sectors of an $n$-quasiparticle system (Theorem \ref{sidethm1}). \\

In Section \ref{anyonsec}, we review the braiding and fusion structure from Section \ref{superbraidsec} within the usual setting of braided $6j$ fusion systems, and present the additional postulates required to make contact with the framework of ribbon fusion categories. In Section \ref{ssbremk1of2}, we observe some $R$-matrix identities that follow from our construction: these reveal some information about the spectrum of $\beta_{n}$, and provide an ansatz for the monodromy operator which is consistent with the categorical ribbon relation.\\

In Section \ref{future}, we give a concise summary of our exposition, and speculate on a possible extension of our construction.

\section{Preliminaries}
\label{prelimsec}
\subsection{Superselection rules and identical particles}
\label{SSRIP}
\noindent Consider a system with Hilbert space $\mathcal{H}$. A \textit{superselection rule} (SSR) is given by a normal operator $\hat{J}:\mathcal{H}\to\mathcal{H}$ where
\begin{equation} [\hat{O},\hat{J}]=0 \label{ssrdefcomm}\end{equation}
for \textit{all} observables $\hat{O}$ of the system. Suppose that $\mathcal{H}'$ and $\mathcal{H}''$ are any two distinct \textit{superselection sectors} (eigenspaces of $\hat{J}$). Then (\ref{ssrdefcomm}) tells us that for any \mbox{$\ket{\psi'}\in\mathcal{H}'$}, $\ket{\psi''} \in\mathcal{H}''$ and any observable $\hat{O}$ on $\mathcal{H}$, we have
\begin{equation} \bra{\psi'}\hat{O}\ket{\psi''}=0 \end{equation}
The defining feature of SSRs is that they preclude the observation of relative phases between states from distinct superselection sectors: let $\ket{\psi}=\alpha\ket{\psi'}+\beta\ket{\psi''}$ and \mbox{$\ket{\psi_{_{\theta}}}=\alpha\ket{\psi'}+e^{i\theta}\beta\ket{\psi''}$} be  normalised states. We have 
\begin{equation}\langle \hat{O} \rangle_{\psi}=\langle \hat{O} \rangle_{\psi_{_{\theta}}}=tr(\hat{O}\hat{\rho}) \ \ \text{for all} \ \hat{O},\theta\end{equation}
where $\hat{\rho}=|\alpha|^{2}\ket{\psi'}\bra{\psi'}+|\beta|^{2}\ket{\psi''}\bra{\psi''}$ (i.e.\ if superpositions $\psi_{_{\theta}}$ were to exist, we would be incapable of physically distinguishing them from a statistical mixture). \\ \\
Examples of superselection observables\footnote{SSRs for which $\hat{J}$ is an observable.} include spin, mass\footnote{Bargmann's mass SSR arises through demanding the Galilean covariance of the Schr\"{o}dinger equation: this only pertains to nonrelativistic systems, since Galilean symmetry is superseded by Poincar\'{e} symmetry in special relativity.} and electric charge. 
\mbox{Notably}, the spin SSR concerns the superposition of integer and half-integer spins: by the spin-statistics theorem, this is equivalent to the boson-fermion SSR. These two equivalent SSRs are sometimes referred to as the univalence SSR. \\  \\
The \textit{intrinsic} properties of a particle may be characterised as corresponding to \mbox{quantum} numbers with an associated SSR. Two particles are \textit{identical} if all of their \mbox{intrinsic} properties match exactly e.g.\ all electrons are identical. 

\subsection{Particle exchanges}
\label{pexchgs}

Consider the exchanges of $n$ identical particles\footnote{It will be assumed that particles are point-like.} on a connected $m$-manifold $\mathcal{M}$ for $m\geq2$. The homotopy classes of exchange trajectories in $\mathcal{M}$ form a group $G_{n}(\mathcal{M})\cong\pi_{1}(\mathcal{U}_{n}(\mathcal{M}))$ under composition (the fundamental group of the $n$th unordered configuration space of $\mathcal{M}$). We will call this the \textit{motion group}. We are interested in two cases for $\mathcal{M}$. Firstly, we have $G_{n}(\mathbb{R}^{d})\cong S_{n}$ (the symmetric group) for $d\geq3$. Here, a tangle\footnote{We call two successive exchanges of the same orientation on a pair of adjacent particles a \textit{tangle}.} is homotopic to $0$ tangles and exchanges are insensitive to orientation (Figure \ref{nulltangle}).\vspace{-3mm}
\vspace{3mm}
\begin{figure}[H]\centering \includegraphics[width=0.75\textwidth]{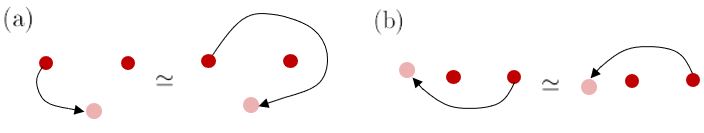} \caption{Exchange trajectories in $\mathbb{R}^{d}$ for (a) a clockwise tangle (right), and \mbox{(b) single exchanges}. When $d\geq3$, deformations `$\simeq$'  lift the strands through the extra spatial dimension(s).} \label{nulltangle}\end{figure}
\vspace{-5mm}
\noindent Secondly, for a 
surface $\mathcal{S}$ we have $G_{n}(\mathcal{S})\cong B_{n}(\mathcal{S})$ (the surface braid group). 
Given any $n$ points in (the interior of) $\mathcal{S}$, we can take some disc $D\subset\mathcal{S}$ such that all $n$ points lie inside $D$. We know that $G_{n}(\mathbb{D}^{2})\cong B_{n}$ where $\mathbb{D}^{2}$ is the $2$-disc and
\begin{equation}B_{n}=\left\langle\begin{array}{c|c}
      \sigma_{1},\ldots,\sigma_{n-1} & \begin{array}{c}
	\sigma_{i}\sigma_{i+1}\sigma_{i}=\sigma_{i+1}\sigma_{i}\sigma_{i+1} \\
	\sigma_{i}\sigma_{j}=\sigma_{j}\sigma_{i} \ , \ |i-j|\geq2
\end{array} \end{array}\right\rangle\end{equation}
is the Artin braid group. We will denote the identity element by $e$. The braid relations for $B_{n}$ thus also hold in $B_{n}(\mathcal{S})$ \cite{guaschi}. When considering particle exchanges on a surface $\mathcal{S}$, we henceforth restrict our attention to $B_{n}(\mathbb{D}^{2})$.\vspace{1mm}
\begin{figure}[H]\centering \includegraphics[width=0.75\textwidth]{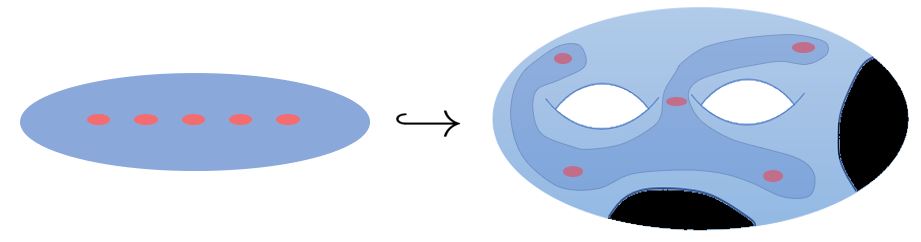} \caption{Particles are considered as lying in some disc $D\subset\mathcal{S}$. Since we are only interested in the topology of exchange trajectories and $B_{n}(\mathbb{D}^{2})\cong B_{n}(D)$, we can restrict our attention to particles in $\mathbb{D}^{2}$.} \label{discinc}\end{figure}\vspace{-2mm}
\begin{remk}In particular, this means that what we learn about the exchange statistics of \mbox{particles} on a disc is also applicable to particles on surfaces with arbitrary \mbox{topology}. \label{arbtopremk}\end{remk}
\vspace{-2mm}
\begin{figure}[H]\centering \includegraphics[width=0.52\textwidth]{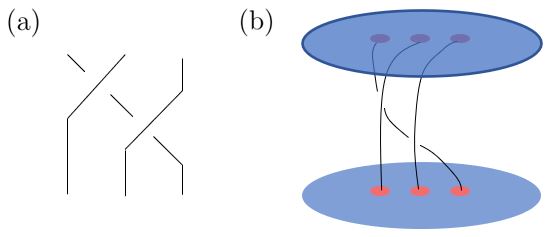} \caption{A braid diagram with $n$ strands will be interpreted as a worldline diagram for $n$ particles on a disc. \textit{We will let the time axis run downwards}. The above diagram depicts this for the $3$-braid $\sigma_{2}\sigma_{1}$.} \label{braidworldlineex}\end{figure}\vspace{2mm}

\vspace{-3mm}
\section{Exchange Symmetry in Three or More Spatial Dimensions}
\label{exches3}
A permutation of $n$ identical particles will be indistinguishable from the original configuration: this is called \textit{exchange symmetry} and may be concisely expressed by
\begin{equation}[\hat{O},\rho(g)]=0 \label{exsym1}\end{equation}
for all observables $\hat{O}$ on $\mathcal{H}$ (the $n$-particle Hilbert space), and all $g$ in the $n$-particle motion group $G$ where $\rho:G\to U(\mathcal{H})$ is the unitary linear representation describing the evolution in $\mathcal{H}$ under the action of $G$. It is easy to see that if (\ref{exsym2}) holds for all generators $g_{i}$ of $G$, then (\ref{exsym1}) follows. 
\begin{equation}[\hat{O},\rho(g_{i})]=0 \label{exsym2}\end{equation}

\noindent Recall that $S_{n}$ is the motion group of $n$ particles in $\mathbb{R}^{d}$ for $d\geq3$. We write
\begin{equation}S_{n}=\left\langle\begin{array}{c|c}
      s_{1},\ldots,s_{n-1} & \begin{array}{c}
	s_{i}^{2}=e \\
	s_{i}s_{i+1}s_{i}=s_{i+1}s_{i}s_{i+1} \\
	s_{i}s_{j}=s_{j}s_{i} \ , \ |i-j|\geq2
\end{array} \end{array}\right\rangle\end{equation}
If $\dim(\mathcal{H})=1$, it is clear that $\rho$ can only be one of $\rho^{\pm}$ where
\begin{equation}\begin{split}\rho^{\pm} : S_{n} &\to U(1)  \\ s_{i} &\mapsto\pm{1} \end{split}\end{equation}
If two identical particles are exchanged and their wavefunction is scaled by $+1$, they are called \textit{bosons}; if their wavefunction is scaled by $-1$, they are called \textit{fermions}.\\

\noindent Letting $\dim(\mathcal{H})>1$, it is consistent to expect that statistical evolutions determined by higher-dimensional representations of the symmetric group should be possible. Such exchange statistics are referred to as parastatistics. However, `paraparticles' have never been observed in nature, and all known fundamental particles may be classified as being either a boson or a fermion. Indeed, the classification of identical particles as being either bosons or fermions is sometimes included as a postulate of quantum mechanics (called the \textit{symmetrisation postulate}). If this postulate is relaxed then it can still be shown (under the pertinent constraints) that the boson-fermion classification will hold \cite{dop1, dop2, mug07}.\\

\noindent In order for (\ref{exsym1}) to be consistent with the symmetrisation postulate, we must levy some restrictions on $\rho$ when $G=S_{n}$ and $\dim(\mathcal{H})>1$. The eigenvalues of $\rho(s_{i})$ belong to a nonempty subset of $\{\pm1\}$. We respectively denote the corresponding eigenspaces (one of which is possibly zero-dimensional) by $\mathcal{H}_{i}^{\pm}$. Since each such eigenspace defines a superselection sector and the $n$ particles are identical (and are thus either all bosons or all fermions by the postulate), $\rho$ must be such that $\mathcal{H}_{i}^{\pm}=\mathcal{H}_{j}^{\pm}$ for all $i,j$. We thus have $\mathcal{H}=\mathcal{H}^{+}\oplus\mathcal{H}^{-}$ (i.e. the subscripts are dropped). Under this restriction, we may thus recover the boson-fermion SSR from (\ref{exsym1}).

\begin{remk}For a system of $n$ bosons or fermions, there is typically no subspace describing a subsystem of $k<n$ particles. This is implicit in the structure of Fock space\footnote{As a consequence of the mass SSR, note that the sectors of Fock space correspond to a SSR for the particle number operator in the nonrelativistic limit.} (here $\mathcal{H}^{(k)}_{(\pm)}$ denotes the space of (anti)symmetric states for $k$ identical particles):
\begin{equation}\mathscr{H}_{\pm}=\mathcal{H}^{(0)}\oplus\mathcal{H}^{(1)}\oplus\mathcal{H}^{(2)}_{\pm}\oplus\mathcal{H}^{(3)}_{\pm}\oplus\ldots\end{equation} 
E.g. $\mathcal{H}^{(2)}_{+}\not\subset\mathcal{H}^{(3)}_{+}.$ For instance, states such as $\frac{1}{\sqrt{2}}(\ket{01}-\ket{10})\in\mathcal{H}_{-}^{(2)}$ do not describe a physical entanglement, since the subsystem for an individual particle is physically inaccessible \cite{zanardi}. This is in contrast to anyonic systems which have a well-defined description of state spaces for particle subsystems (since anyons are \textit{localised} phenomena). Nonetheless, there exist circumstances under which some notion of distinguishability amongst $n$ identical bosons or fermions may be recovered: for instance, when their wavefunctions have (approximately) disjoint compact support. This can happen if the particles are far apart, or separated by sufficiently strong potentials.\end{remk}  

\section{Exchange Symmetry in Two Spatial Dimensions}
\label{exches2}
\subsection{Quasiparticles and braiding}
\label{quab}
Although there are no fundamental particles in two spatial dimensions, it is well-known that various two-dimensional systems are theoretically capable of supporting localised excitations with fractional statistics \cite{mooreread,xwenfqhe,parafermions,weylspinliquids}: these emergent phenomena are known as \textit{quasiparticles}\footnote{We will use the terms `quasiparticle' and `particle' interchangeably.}; they have no internal degrees of freedom and may thus be considered as \textit{identical}. The localised nature of these excitations is instrumental in the emergence of fusion structure. \\ \\
\fbox{
\parbox{\textwidth}{
\textbf{A1.} Two-dimensional quasiparticles are \textit{spatially localised} phenomena.}} \\ \\

Recall that $B_{n}$ is the motion group of $n$ particles on a disc. Then for a two-quasiparticle system with Hilbert space $\mathcal{V}$, the action of the motion group is is given by a unitary linear representation 
\begin{equation}\rho:B_{2}\to U(\mathcal{V})\label{saskexch}\end{equation}
The eigenvalues $\{e^{iu_{_{X}}}\}_{X}$ of $\rho(\sigma_{1})$ lie in $U(1)$, and we have the corresponding decomposition $\mathcal{V}=\bigoplus_{X}\mathcal{V}_{X}$ where eigenspaces $\mathcal{V}_{X}$ define superselection sectors by exchange symmetry as expressed in (\ref{exsym1}).\footnote{It is assumed that the superselection sectors are finite-dimensional, and that the number of distinct superselection sectors is finite. This assumption is later codified as \textbf{A2} in Section \ref{finduel}.} The possibly arbitrary exchange phase $e^{iu_{X}}$ is what earns \textit{any}ons their namesake \cite{wilczek}. 
\begin{remk}Now consider an $n$-particle system for $n\geq2$. \textbf{A1} permits us to consider the Hilbert space \mbox{associated} with a subsystem of $k$ adjacent quasiparticles (where $2\leq k\leq n$). Consequently, the action of the \mbox{motion} \textit{subgroup} $B_{k}$ on any such subsystem will be independent of the rest of the \mbox{system}. The description of the superselection sectors and exchange statistics given by the action of $B_{2}$ on some pair of quasipartilces is thus a property intrinsic to said pair.\label{a1cons1}\end{remk}
\noindent Consider a $2$-quasiparticle subsystem (of particles labelled $q_{i}$ and $q_{i+1}$ located at the $i^{th}$ and ${i+1}^{th}$ positions respectively) of an $n$-quasiparticle system. Denote the Hilbert space of this subsystem by $\mathcal{V}^{\{q_{i}, q_{i+1}\}}$ where $\{q_{i},q_{i+1}\}$ is an unordered set. Following Remark \ref{a1cons1}, (\ref{saskexch}) defines a fixed action
\begin{equation}\rho_{\{q_{i},q_{i+1}\}}:B_{2}\to U(\mathcal{V}^{\{q_{i},q_{i+1}\}})\label{subfixie}\end{equation}
on $q_{i}$ and $q_{i+1}$, and we write the eigenspace decomposition $\mathcal{V}^{\{q_{i},q_{i+1}\}}=\bigoplus_{X}\mathcal{V}^{\{q_{i},q_{i+1}\}}_{X}$ for $\rho_{\{q_{i},q_{i+1}\}}(\sigma_{1})$. We label the quasiparticles from \mbox{$1$ to $n$} and let $S_{\{1,\ldots, n\}}$ be the set whose \mbox{elements} are all possible permutations of the string $12\ldots n$. Given some $s\in S_{\{1,\ldots,n\}}$ we write $s=q_{1}\ldots q_{n}$ where $q_{i}$ is the $i^{th}$ character of string $s$. We denote the Hilbert space for quasiparticles $q_{1}\ldots q_{n}$ (in that order) by $V^{q_{1}\ldots q_{n}}$ or $V^{s}$. E.g.\ $V^{q_{1}\ldots q_{i}q_{i+1}\ldots q_{n}}$ and $V^{q_{1}\ldots q_{i+1}q_{i}\ldots q_{n}}$ are the state spaces assigned to the system in the initial and final time-slices of Figure \ref{subsysdiag} respectively. 

\begin{figure}[H]\centering \includegraphics[width=0.4\textwidth]{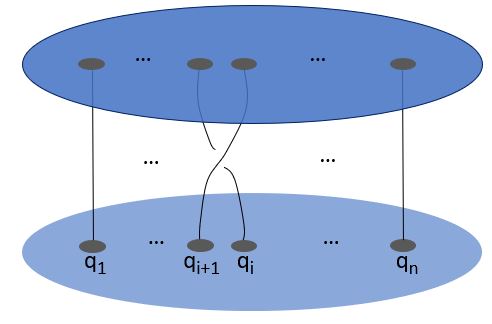} \caption{The clockwise exchange of quasiparticles $q_{i}$ and $q_{i+1}$.} \label{subsysdiag}\end{figure}

\vspace{3.5mm}

\noindent  Let $\rho_{s}\big\vert_{V^{s}}(\sigma_{i})$ be the unitary linear transformation describing the action of braid $\sigma_{i}\in B_{n}$ on the \mbox{$n$-quasiparticle} system (as shown in Figure \ref{subsysdiag}). For $n>2$, 
\begin{equation}V^{q_{1}\ldots q_{i}q_{i+1}\ldots q_{n}}\cong V^{q_{1}\ldots q_{i+1}q_{i}\ldots q_{n}} \cong \bigoplus_{X}\mathcal{V}^{\{q_{i},q_{i+1}\}}_{X}\otimes\bar{V}^{(s)}_{X} \label{diwalimazurka}\end{equation}
where $\bar{V}^{(s)}_{X}$ denotes the state space for the rest of the system when $q_{i}$ and $q_{i+1}$ are in superselection sector $X$. The spaces $V^{q_{1}\ldots q_{i}q_{i+1}\ldots q_{n}}$ and $V^{q_{1}\ldots q_{i+1}q_{i}\ldots q_{n}}$ may be identified under the action of the subgroup $\langle \sigma_{1},\ldots,\sigma_{i-2},\widehat{\sigma}_{i-1},\sigma_{i},\widehat{\sigma}_{i+1},\sigma_{i+2},\ldots\sigma_{n-1}\rangle$, but are only equivalent up to isomorphism under the action of $B_{n}$. This is because the action of $B_{n}$ on the system will \mbox{generally} depend upon the order of the quasiparticles for $n>2$. E.g.\ the action of $\sigma_{1}\in B_{3}$ on $V^{123}$ will differ from its action on $V^{231}$ (unless $\rho_{\{1,2\}}$ and $\rho_{\{2,3\}}$ are the same). We must therefore distinguish between the spaces $\{V^{s}\}_{s\in S_{\{1,\ldots, n\}}}$ in order to consider the action of braiding on the whole system.  
\begin{equation}\rho_{s}\big\vert_{V^{s}}(\sigma_{i}^{\pm1})=\bigoplus_{X}\left[\rho^{X}_{\{q_{i}, q_{i+1}\}}(\sigma_{1}^{\pm1})\otimes \text{id}_{\bar{V}^{(s)}_{X}}\right]\label{submotextrap}\end{equation}
where $\rho^{X}_{\{q_{i}, q_{i+1}\}}$ is the subrepresentation given by restricting $\rho_{\{q_{i}, q_{i+1}\}}$ to $\mathcal{V}^{\{q_{i}, q_{i+1}\}}_{X}$. 

\vspace{5mm}
\begin{defn}
$\rho_{s}(\sigma_{i}^{\pm1})$ denotes the action of (anti)clockwise exchanging $q_{i}$ and $q_{i+1}$ on an $n$-particle Hilbert space. It is therefore \textit{necessary that} $u\in S_{\{1,\ldots,n\}}$ contains the substring $q_{i}q_{i+1}$ or $q_{i+1}q_{i}$ for any $V^{u}$ on which $\rho_{s}(\sigma_{i}^{\pm1})$ is defined. Following from (\ref{submotextrap}), that is
\begin{equation}\rho_{s}\big\vert_{V^{u}}(\sigma_{i}^{\pm1})=\bigoplus_{X}\left[\rho^{X}_{\{q_{i}, q_{i+1}\}}(\sigma_{1}^{\pm1})\otimes \text{id}_{\bar{V}^{(u)}_{X}}\right]\label{submotextrap2}\end{equation}
\label{babylonburns}
\end{defn}




\vspace{3mm}
The above tells us that the right way to think about the action of braiding on an $n$-quasiparticle system is as follows: let $\{V^{s}\}_{s}$ be defined as above and let $b(s)\in S_{\{1,\dots n\}}$ be the obvious permutation\footnote{E.g. $\sigma_{i}^{\pm1}(q_{1}\ldots q_{i}q_{i+1}\ldots q_{n})=q_{1}\ldots q_{i+1}q_{i}\ldots q_{n}$ \ . That is, $b(s)$ is the string obtained by reading off the labels of the endpoints of braid $b$ when its starting points are labelled (left-to-right) by the characters of $s$.} of $s$ for any $b\in B_{n}$ . We construct an action of the braids $b\in B_{n}$ as linear transformations between spaces $\{V^{s}\}_{s}$. This action is defined through a collection of functions $\{\rho_{s}\}_{s}$ such that \ref{b0}-\ref{b5} hold for any $s\in S_{\{1,\dots n\}}$ and for all $b,b_{1},b_{2}\in B_{n}.$

\newpage

\fbox{
\parbox{\textwidth}{
\begin{enumerate}[start=0,label=\textbf{(B\arabic*)}]
\item The domain of $\rho_{s}$ is the braid group $B_{n}$  \label{b0}
\item The image of $b$ under $\rho_{s}$ is a linear transformation
\begin{equation}\rho_{s}(b):\bigoplus_{u\in\mathcal{U}_{s,b}}V^{u}\to\hspace{-2mm}\bigoplus_{s'\in S_{\{1,\ldots,n\}}}\hspace{-3mm}V^{s'}\end{equation}
where the elements $u\in\mathcal{U}_{s,b}\subseteq S_{\{1,\ldots,n\}}$ index the direct summands \\ $\{V^{u}\}_{u}\subseteq\{V^{s'}\}_{s'}$ that constitute the domain of $\rho_{s}(b)$. We have 
\begin{equation}\mathcal{U}_{s,e}:=S_{\{1,\ldots,n\}}\label{lasulaa}\end{equation}
\label{b1}
\vspace{-4mm}
\item For any $u\in\mathcal{U}_{s,b}$, we have linear isomorphism
\begin{equation}\rho_{s}\big\vert_{V^{u}}(b):V^{u}\xrightarrow{\sim}V^{b(u)}\end{equation}
and if $u'\notin\mathcal{U}_{s,b}$ then $\rho_{s}(b)$ is undefined on $V^{u'}$.\vspace{3mm}
\label{b2}
\item Given $b$ such that $b=b_{2}b_{1}$, then for any $u$ such that  $u\in\mathcal{U}_{s,b_1}$ and\\
$b_{1}(u)\in\mathcal{U}_{b_{1}(s),b_{2}}$, we have
\begin{equation}\rho_{s}\big\vert_{V^{u}}(b_{2}b_{1})=\rho_{b_{1}(s)}\big\vert_{V^{b_{1}(u)}}(b_{2})\circ\rho_{s}\big\vert_{V^{u}}(b_{1})\end{equation}
\label{b3}
\item $\rho_{s}\big\vert_{V^{u}}(b)$ is a unitary transformation i.e.\ for $u\in\mathcal{U}_{s,b}$ the map $\rho_{s}\big\vert_{V^{u}}(b)$ \\
has Hermitian adjoint
\begin{equation}\left(\rho_{s}\big\vert_{V^{u}}(b)\right)^{\dagger}=\rho_{b(s)}\big\vert_{V^{b(u)}}(b^{-1})\end{equation}
where
\begin{subequations}\begin{align}
\rho_{b(s)}\big\vert_{V^{b(u)}}(b^{-1}) \ \circ \ \rho_{s}\big\vert_{V^{u}}(b) &= \text{id}_{V^{u}}=\rho_{s}\big\vert_{V^{u}}(e) \\
\rho_{s}\big\vert_{V^{u}}(b) \ \circ \ \rho_{b(s)}\big\vert_{V^{b(u)}}(b^{-1}) &= \text{id}_{V^{b(u)}}=\rho_{b(s)}\big\vert_{V^{b(u)}}(e)
\end{align}\end{subequations}
\label{b4}
\item $\rho_{s}\big\vert_{V^{u}}(\sigma_{i}^{\pm1})$ is defined as in Definition \ref{babylonburns} for $u\in\mathcal{U}_{s,\sigma_{i}^{\pm1}}$
\label{b5}
\end{enumerate}
}
}

\vspace{3mm}

\noindent Let us unpack some details. Firstly, what constitutes $\mathcal{U}_{s,b}$? \ref{b4} tells us that $\rho_{s}\big\vert_{V^{u}}(b)$ is invertible\footnote{In fact, it tells us that $\rho_{s}\big\vert_{V^{u}}(b)$ is a diagonalisable, norm-preserving map for any $u\in\mathcal{U}_{s,b}$.}, whence we must have 
\begin{equation}u\in\mathcal{U}_{s,b}\iff b(u)\in\mathcal{U}_{b(s),b^{-1}}\label{akyonuma}\end{equation} 
and therefore
\begin{equation}u\in\mathcal{U}_{s,\sigma_{i}^{\pm1}}\iff \sigma_{i}(u)\in\mathcal{U}_{\sigma_{i}(s),\sigma_{i}^{\mp1}}\label{ukyonamo}\end{equation}
Combining this with \ref{b5}, we deduce that $\mathcal{U}_{s,\sigma_{i}}$ contains all $u\in S_{\{1,\ldots,n\}}$ such that
\begin{enumerate}[label=(\roman*)]
\item $u$ contains the substring $q_{i}q_{i+1}$ or $q_{i+1}q_{i}$
\item $u$ satisfies (\ref{ukyonamo})
\end{enumerate}
That is, $\mathcal{U}_{s,\sigma_{i}}$ contains all $u$ such that $u$ contains the substring $q_{i}q_{i+1}$ or $q_{i+1}q_{i}$, and for which said substring does not begin at the $i-1^{th}$ or $i+1^{th}$ character of $u$.\\

\vspace{-2mm}
\noindent  Clearly, $\mathcal{U}_{s,\sigma_{i}}=\mathcal{U}_{s,\sigma_{i}^{-1}}$. \ref{b3} tells us that if $u\in\mathcal{U}_{s,b_{1}}$ and $b_{1}(u)\in\mathcal{U}_{b_{1}(s),b_{2}}$, then $u\in\mathcal{U}_{s,b_{1}b_{2}}$. One can check that \ref{b3} together with (\ref{ukyonamo}) yields (\ref{akyonuma}) as required. Also, by combining \ref{b3} with our knowledge of $\mathcal{U}_{s,\sigma_{i}^{\pm1}}$, we can\mbox{ find $\mathcal{U}_{s,b}$.}\footnote{Note that in this construction of $\mathcal{U}_{s,b}$, one considers all braid words of minimal length for $b$.}
 
\begin{remk}\textbf{(Well-definedness and existence)}\\
We know that for $u\in\mathcal{U}_{s,b}$, the map $\rho_{s}\big\vert_{V^{u}}(b)$ may be parsed into a composition of maps of the form in (\ref{submotextrap2}). When $n\geq3$, there exist braids $b$ for which there is more than one way to write $b$ as a product of generators (i.e. as a braid word). This results in $\rho_{s}\big\vert_{V^{u}}(b)$ being given by distinct compositions. In order for the action $\{\rho_{s}\}_{s}$ to be well-defined, we require that all distinct compositions for a given $\rho_{s}\big\vert_{V^{u}}(b)$ are equal.\footnote{Of course, there are cases where two distinct compositions may `automatically' be equal by commutativity of constituent maps.} This intricate requirement is known as a \textit{coherence condition}: we later see that it is fulfilled by demanding that matrix representations for maps of the form (\ref{submotextrap2}) satisfy the so-called \textit{hexagon equations} (see Remark \ref{cocoremk}). For the current purposes of our construction, we will just assume that such (nontrivial) actions (satisfying this coherence condition) exist. 
\end{remk}
\vspace{3mm}
 
\noindent For any map $\rho_{s}\big\vert_{V^{u}}(b)$, we have 
\begin{equation} \rho_{s}\big\vert_{V^{u}}(b)=\rho_{s}\big\vert_{V^{u}}(b\cdot e)=\rho_{s}\big\vert_{V^{u}}(b)\circ \rho_{s}\big\vert_{V^{u}}(e) \end{equation}
whence it is clear that (\ref{lasulaa}) must hold and that $\rho_{s}\big\vert_{V^{u}}(e)=\text{id}_{V^{u}}$. Also note that we always have $s\in\mathcal{U}_{s,b}$, and so we may write
\begin{equation}\rho_{s}\big\vert_{V^{s}}:B_{n}\to\text{Hom}\left(V^{s},\bigoplus_{s'\in S_{\{1,\ldots,n\}}}V^{s'}\right)\end{equation} 
\noindent where $\rho_{s}\big\vert_{V^{s}}(b):V^{s}\xrightarrow{\sim}V^{b(s)}$ is a unitary linear transformation.\\

\noindent Take any $b\in B_{n}$ whose image under the epimorphism $\eta: B_{n}\to S_{n}$ (whose kernel is the normal subgroup $PB_{n}$ of $n$-strand pure braids) is a permutation of the form
\begin{equation*}\begin{pmatrix}
1 & \cdots & i-1 & i & i+1 & i+2 & \cdots & n \\ b(1) & \cdots & b(i-1) & j & j+1 & b(i+2) & \cdots & b(n)
\end{pmatrix}\end{equation*}
or
\begin{equation*}\begin{pmatrix}
1 & \cdots & i-1 & i & i+1 & i+2 & \cdots & n \\ b(1) & \cdots & b(i-1) & j+1 & j & b(i+2) & \cdots & b(n)
\end{pmatrix}\end{equation*}
Then for all $u\in\mathcal{U}_{s,\sigma_{i}}\cap\mathcal{U}_{b(s),\sigma_{j}}$,
\begin{equation}\rho_{s}\big\vert_{V^{u}}(\sigma_{i}^{\pm1})=\rho_{b(s)}\big\vert_{V^{u}}(\sigma_{j}^{\pm1})\label{transhit}\end{equation}

\vspace{6mm}

\noindent The above construction for the ``action'' $\{\rho_{s}\}_{s}$ of $n$-braids on the spaces $\{V^{s}\}_{s}$ can be thought of as a \textit{unitary linear representation of the braid groupoid for $n$ distinctly coloured strands}. A further discussion of this statement is provided in Appendix \ref{cbgactappx}.

\subsection{Exchange symmetry for $n$ quasiparticles}
\label{exches2dsec}

\noindent Recall that superselection sectors arise from exchange symmetry as in (\ref{exsym1}). A subtle but crucial point in this equation is that the $n$-particle Hilbert space does not depend on the order of the particles. This necessity becomes clearer when we try to write down a (\textit{naive}) version of (\ref{exsym1}) compatible with the braiding action described above: 
\[ [\hat{O},\rho_{s}(b)]=0 \text{ \ for all \ } s\in S_{\{1,\ldots, n\}}  \text{ \ and all \ } b\in B_{n} \] 
For starters, the image of $\rho_{s}(b)$ could be in any one of the spaces $\{V^{s}\}_{s}$, so the space of observables should be defined on the $n$-particle Hilbert space ``modulo ordering''. Let us denote such a space by $\mathcal{V}^{[n]}$ where $[n]:=\{1,\ldots,n\}$ is an unordered set. This also makes sense physically, since we should not have different sets of observables depending on the order of the particles (by indistinguishability). This also excludes observables defined on subsystems (which is desirable as we want to consider the exchange symmetry mechanism local to all $n$ quasiparticles). However, in order for the commutator to be well-defined, the braiding action must also be defined on $\mathcal{V}^{[n]}$.\\

\noindent Altogether, the correct adaptation of (\ref{exsym1}) should be given by a commutator of the form
\begin{equation}[\hat{O},\rho_{[n]}(g)]=0\label{truqpexchsym_intro}\end{equation}
for all $n$-particle observables $\hat{O}$ defined on $\mathcal{V}^{[n]}$ and for all $g\in\mathcal{E}_{n}\leq B_{n}$, where
\begin{equation}\rho_{[n]}:\mathcal{E}_{n}\to U(\mathcal{V}^{[n]})\label{vengahell_intro}\end{equation}
is some unitary linear representation.\\

\noindent At first, this formulation of exchange symmetry appears rather abstract. In order to obtain a better understanding of what is meant by (\ref{truqpexchsym_intro})-(\ref{vengahell_intro}), we will outline their construction from the action $\{\rho_{s}\}_{s}$. Take $\mathcal{E}_{n}$ to be the subset of $n$-braids such that for any $g\in\mathcal{E}_{n}$, $b\in B_{n}$ and $s\in S_{\{1,\ldots,n\}}$, we have
\begin{equation}\rho_{b(s)}(g)\cdot\rho_{s}(b)\ket{\psi}=e^{iu_{Q}}\ket{\psi} \label{heurozone_intro}\end{equation}
where $V^{s}=\bigoplus_{Q} V^{s}_{Q}$ is the eigenspace decomposition under the unitary operator $\rho_{s}(g)$ with $\rho_{s}(g)\ket{\psi}=e^{iu_{Q}}\ket{\psi}$ for any $\ket{\psi}\in V^{s}_{Q}$. Equation (\ref{heurozone_intro}) demands that $V^{s}_{Q}$ is the $e^{iu_{Q}}$-eigenspace of $\rho_{s}(g)$ for all $s$. Since the $e^{iu_{Q}}$-eigenspace (of the action of $g\in\mathcal{E}_{n}$) is stable under the action of all $n$-braids, it is independent of the order of the particles: we thus denote it by $\mathcal{V}^{[n]}_{Q}$ where $\mathcal{V}^{[n]}=\bigoplus_{Q}\mathcal{V}^{[n]}_{Q}$. The action of $g$ on $\mathcal{V}^{[n]}$ is denoted by $\rho_{[n]}(g)$ where
\begin{equation}\rho_{[n]}(g)=\sum_{Q}e^{iu_{Q}}\hat{P}_{Q}\end{equation}
and where $\hat{P}_{Q}$ is a normalised projector onto $\mathcal{V}^{[n]}_{Q}$. We will see that $\mathcal{E}_{n}$ is a subgroup generated by a single $n$-braid i.e. 
\begin{equation}\mathcal{E}_{n}=\langle \beta_{n} \rangle \leq B_{n} \label{lonelyexchgrp}\end{equation}
and we will therefore call $\beta_{n}\in B_{n}$ the \textit{superselection braid}. The $n$-quasiparticle Hilbert space ``modulo ordering'' may therefore be understood as the representation space in (\ref{vengahell_intro}), which in turn is constructed through the action $\{\rho_{s}\}_{s}$ of $n$-braids on the spaces $\{V^{s}\}_{s}$.\footnote{Recall from (\ref{submotextrap2}) that the action $\{\rho_{s}\}_{s}$ can be formulated in terms of the pairwise action (\ref{subfixie}). The $n$-fold exchange symmetry mechanism (\ref{truqpexchsym_intro}) may thus be thought of as emerging from the pairwise exchange symmetries among its constituents.} From the above, it is clear that $\mathcal{V}^{[n]}_{Q}\cong V^{s}_{Q}$ for any $s$. \\

\noindent Since the action of the superselection $n$-braid does not depend on the order of the particles, the braid itself should not favour any single particle over another. This hints that the braid should realise ${n \choose 2}$ exchanges (i.e.\ each pair is exchanged once).\\ By the innate symmetry of the representation space $\mathcal{V}^{[n]}$, we expect that the braid word $\beta_{n}$ should also satisfy several internal symmetries. Indeed, we will subsequently see that these properties are satisfied, and that the superselection braid is given by
\begin{equation}\beta_{n}=\sigma_{1}\sigma_{2}\ldots\sigma_{n-1}\cdot \sigma_{1}\sigma_{2}\ldots\sigma_{n-2}\cdot\ldots\cdot\sigma_{1} \ \ , \quad n\geq2\label{ssbprop_intro}\end{equation}
and $\beta_{1}=e$. Studying the action of this braid reveals the fusion structure amongst quasiparticles and hints at their topological spin structure. This is key in connecting the narrative of exchange symmetry to the framework of braided fusion categories.

\subsection{Superselection sectors for $n$ quasiparticles}
\label{supsec2dsec}
\noindent  Given a system $V^{q_{1}\ldots q_{n}}$ of $n>2$ quasiparticles, note that exchange symmetry (\ref{truqpexchsym_intro}) is defined with respect to all subsystems of $k$ adjacent quasiparticles (where $2\leq k \leq n$) i.e.
\begin{equation}[\hat{O},\rho_{[k]}(\beta_{k})]=0\label{truqpexchsym2}\end{equation}
for all observables $\hat{O}$ on $\mathcal{V}^{[k]}$.\footnote{In the instance of subsystems, $[k]$ denotes the unordered set of labels for the $k$ particles.} We therefore have a hierarchy of exchange symmetries. The next step is to understand how these all fit together. Equation (\ref{truqpexchsym2}) tells us that the eigenspaces $\{\mathcal{V}^{[k]}_{X}\}_{X}$ of $\rho_{[k]}(\beta_{k})$ define superselection sectors.
\ Take the $k$-particle subsystem $V^{q_{1}\ldots q_{k}}$ and write the decomposition into $k$-particle superselection sectors as $\bigoplus_{X}V^{q_{1}\ldots q_{k}}_{X}=V^{q_{1}\ldots q_{k}}$. \\

\noindent ($\mathcal{Q}$) \textit{ How are $\{V^{q_{1}\ldots q_{k}}_{X}\}_{X}$ understood in the context of the full $n$-particle system? } \\ 

\noindent Let $k<n$ and write the decomposition into $n$-particle superselection sectors as $\bigoplus_{Q}V^{q_{1}\ldots q_{n}}_{Q}=V^{q_{1}\ldots q_{n}}$. Suppose the $n$-particle state is in superselection sector $V^{q_{1}\ldots q_{n}}_{Q}$. The most general way to decompose $V^{q_{1}\ldots q_{n}}_{Q}$ with respect to the $k$-particle subsystem is
\begin{equation}V^{q_{1}\ldots q_{n}}_{Q}\cong\bigoplus_{X} V^{q_{1}\ldots q_{k}}_{X} \otimes V^{X,q_{k+1}\ldots q_{n}}_{Q}\label{diwalichoro_intro}\end{equation}
where $V^{X,q_{k+1}\ldots q_{n}}_{Q}$ denotes the state space for the rest of the system when $q_{1},\ldots,q_{k}$ are in superselection sector $X$. \\

\noindent Let us compare (\ref{diwalimazurka}) and (\ref{diwalichoro_intro}) when $i=1$ and $k=2$. In this case, $V^{q_{1}q_{2}}_{X}\cong \mathcal{V}^{\{q_{1},q_{2}\}}_{X}$ and $\bigoplus_{Q} V^{X,q_{k+1}\ldots q_{n}}_{Q}\cong\bar{V}^{(s)}_{X}$. Spaces $V^{q_{1}q_{2}}_{X}$ and $V^{q_{2}q_{1}}_{X}$ may be identified with $\mathcal{V}^{\{q_{1},q_{2}\}}_{X}$ when considered as representation spaces of $B_{2}$, but are distinguished between in the context of a larger system (since we usually need to keep track of the particle ordering) and are thus only considered equivalent up to isomorphism.\\

\noindent We can also partition an $n>3$ particle system into subsystems $V^{q_{1}\ldots q_{k}}$ and $V^{q_{k+1}\ldots q_{n}}$ where we assume $2\leq k \leq n-2$. Denote the superselection sectors of each by $\{V^{q_{1}\ldots q_{k}}_{X}\}_{X}$ and $\{V^{q_{k+1}\ldots q_{n}}_{Y}\}_{Y}$. Suppose the $n$-particle state is in superselection sector $V^{q_{1}\ldots q_{n}}_{Q}$. The most general way to decompose $V^{q_{1}\ldots q_{n}}_{Q}$ with respect to the two subsystems is
\begin{equation}V^{q_{1}\ldots q_{n}}_{Q}\cong\bigoplus_{X,Y} V^{q_{1}\ldots q_{k}}_{X} \otimes V^{XY}_{Q} \otimes V^{q_{k+1}\ldots q_{n}}_{Y}\label{wineup_intro}\end{equation}
The spaces $\{V^{XY}_{Q}\}_{X,Y}$ may be thought of as constraining the superselection sectors of the subsystems by relating them to the $n$-particle superselection sector.\\ If $\dim(V^{XY}_{Q})=d$, this may be interpreted as the superselection sector $Q$ containing superselection sectors $X$ and $Y$ in ``$d$ distinct ways''. We may have $d=0$, but it is also clear that at least one of the spaces $\{V^{XY}_{Q}\}$ must be nonzero.\footnote{This is equivalent to saying at least one of the spaces $\{V^{X,q_{k+1}\ldots q_{n}}_{Q}\}_{X}$ must be nonzero in (\ref{diwalichoro_intro}).} By comparing (\ref{diwalichoro_intro}) and (\ref{wineup_intro}), we see that
\begin{equation}V^{X,q_{k+1}\ldots q_{n}}_{Q}\cong\bigoplus_{Y}V^{XY}_{Q}\otimes V^{q_{k+1}\ldots q_{n}}_{Y}\end{equation}
Analogously to (\ref{diwalichoro_intro}) we can write $V^{q_{1}\ldots q_{n}}_{Q}=\bigoplus_{Y} V^{q_{1}\ldots q_{k},Y}_{Q} \otimes V^{q_{k+1}\ldots q_{n}}_{Y}$ whence it similarly follows that 
\begin{equation}V^{q_{1}\ldots q_{k},Y}_{Q}\cong\bigoplus_{X} V^{q_{1}\ldots q_{k}}_{X} \otimes V^{XY}_{Q}\end{equation}
\noindent In light of the above, it is easy to check that a ``$1$-quasiparticle Hilbert space'' must be canonically isomorphic to $\mathbb{C}$. It is therefore standard practice to omit a $1$-quasiparticle Hilbert space in a decomposition.

\begin{remk}\textbf{(Superselection sectors of subsystems)}\hspace{2mm}\\
\noindent Another salient feature emerges from the hierarchy of superselection sectors in system of $n$ quasiparticles for $n>2$. To illustrate this, consider decomposition (\ref{diwalichoro_intro}). While the spaces $\{V^{q_{1}\ldots q_{k}}_{X}\}_{X}$ still define superselection sectors locally (i.e. with respect to the $k$-particle subsystem), they do not define superselection sectors in the context of the larger system.\footnote{When we look at the whole system ``from afar'' we expect it to be in the ground state. This means that the superselection sector of the whole system should correspond to the vacuum, which later motivates the notion of ``dual charges''.} This is because the $k$-particle exchange symmetry mechanism is superseded by the $n$-particle mechanism. Indeed, the superselection sectors of the subsystem are entangled with the rest of the system in (\ref{diwalichoro_intro}).\footnote{Specifically, when $X$ runs over $>1$ index and at least two of the spaces  $\{V^{X,q_{k+1}\ldots q_{n}}_{Q}\}_{X}$ are nonzero.} Crucially, this means that when we consider the entire system, it is possible to observe linear superpositions over the spaces $\{V^{q_{1}\ldots q_{k}}_{X}\}_{X}$. It is also possible that interactions between the subsystem and the rest of the system induce transitions between superselection sectors of the subsystem.
\label{supersubdestruct}
\end{remk}

\section{The Superselection Braid and Fusion Structure}
\label{superbraidsec}
In Section \ref{exches2dsec}, we outlined the method for determining the superselection sectors using the action $\{\rho_{s}\}_{s}$. The first task is to find the subset $\mathcal{E}_{n}$ of all $n$-braids satisfying (\ref{heurozone_intro}). For any candidate braid $g\in B_{n}$, it suffices to check that (\ref{heurozone_intro}) is satisfied by $b=\sigma_{i}^{\pm1}$ for all $i$. \noindent It will be convenient to define the following notation for braids:
\begin{equation}\sigma_{i_{1}\ldots i_{k-1}i_{k}}:=\sigma_{i_{1}}\ldots\sigma_{i_{k-1}}\sigma_{i_{k}} \ \ , \ \ b_{j}:=\sigma_{12\ldots j} \ \ \text{ for all $j\geq1$, and} \ b_{0}:=e\end{equation}
We argued that a reasonable heuristic for an element of $\mathcal{E}_{n}$ would be that it exchanges each pair of quasiparticles once. Take the ansatz
\begin{equation}\beta_{n}=b_{n-1}b_{n-2}\ldots b_{1} \ \ , \quad n\geq2\label{ssbprop}\end{equation}
E.g.\ $\beta_{2}=\sigma_{1}, \beta_{3}=\sigma_{121}, \beta_{4}=\sigma_{123121}$ etc. We also set $\beta_{1}:=e$. In Theorem \ref{main1}, we will show that $\beta_{n}\in\mathcal{E}_{n}$. Therefore, (the action of) $\beta_{n}$ specifies the superselection sectors; in fact, it does so uniquely up to orientation (Theorem \ref{sidethm1}) which proves (\ref{lonelyexchgrp}) i.e. $\mathcal{E}_{n}=\langle \beta_{n} \rangle\leq B_{n}$. For this reason, we  will refer to $\beta_{n}$ as the \textit{superselection braid}.

\subsection{The superselection braid} 
\label{supselbraidsec}

\begin{thm}{\normalfont\textbf{(Superselection sectors)}}\\ 
We have the eigenspace decomposition $V^{s}=\bigoplus_{Q}V^{s}_{Q}$ under $\rho_{s}(\beta_{n})$ where
\begin{equation}\begin{split}\rho_{s}(\beta_{n}):V^{s}_{Q}&\to V^{\beta_{n}(s)}_{Q} \\
\ket{\Psi}&\mapsto e^{iu_{Q}}\ket{\Psi}\end{split}\ , \ n\geq2\end{equation}
for any $s\in S_{\{1,\ldots,n\}}$.
\label{main1}\end{thm}
\begin{figure}[H]\centering \includegraphics[width=0.28\textwidth]{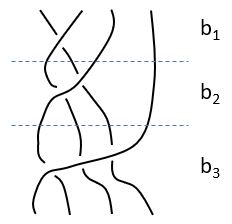} \caption{$\beta_{n}$ has length ${n\choose2}$. The above diagram depicts $\beta_{4}$.} \label{beta4}\end{figure}

\vspace{5mm}
\noindent Let us recap the rest of the construction from Section \ref{exches2dsec}. Theorem \ref{main1} allows us to identify the spaces $\{V^{s}_{Q}\}_{s}$ as the $e^{iu_{Q}}$-eigenspace $\mathcal{V}^{[n]}_{Q}$ under the action of $\beta_{n}$. Write,
\begin{equation}\mathcal{V}^{[n]}=\bigoplus_{Q}\mathcal{V}^{[n]}_{Q}\end{equation}
In particular, this corresponds to a unitary representation
\begin{equation}\rho_{[n]}:\langle\beta_{n}\rangle\leq B_{n}\to U(\mathcal{V}^{[n]})\end{equation}
where
\begin{equation}\begin{split}\rho_{[n]}(\beta_{n}):\mathcal{V}_{Q}^{[n]}&\to\mathcal{V}_{Q}^{[n]}\\
\ket{\varphi}&\mapsto e^{iu_{Q}}\ket{\varphi}\end{split}\end{equation}
That is,
\begin{equation}\rho_{[n]}(\beta_{n})=\sum_{Q}e^{iu_{Q}}\hat{P}_{Q}\end{equation}
where $\hat{P}_{Q}$ is a normalised projector onto $\mathcal{V}^{[n]}_{Q}$. Since the representation space $\mathcal{V}^{[n]}$ is the $n$-quasiparticle Hilbert space (modulo ordering), exchange symmetry is given by
\begin{equation}[\hat{O},\rho_{[n]}(\beta_{n})]=0\label{exsym5}\end{equation}
for all $n$-particle observables $\hat{O}$ on $\mathcal{V}^{[n]}$. The spaces $\{\mathcal{V}^{[n]}_{Q}\}_{Q}$ are superselection sectors of the system, and we have shown by construction that each superselection sector is preserved under the action of of any $n$-braid. It follows that $V^{s}_{Q}$ defines a super-selection sector for any $(s, Q)$. In conclusion, the superselection sectors of an $n$-quasiparticle system are given by the eigenspaces of the action of the braid $\beta_{n}$.


\newpage
\begin{cor}
Given $\ket{\Psi}\in V^{s}_{Q}$ as in Theorem \ref{main1}, 
\begin{equation}\rho_{s}(\beta_{n}^{-1})\ket{\Psi}=e^{-iu_{Q}}\ket{\Psi}\end{equation}
\end{cor}

\begin{proof}Let $\tilde{s}:=\beta_{n}(s)$ (i.e.\ string $s$ in reverse order). By Theorem \ref{main1}, 
\begin{align*}\rho_{\tilde{s}}(\beta_{n})\left[\rho_{s}(\beta_{n})\ket{\Psi}\right]&=e^{iu_{Q}}\left[\rho_{s}(\beta_{n})\ket{\Psi}\right] \\
\implies\rho_{\tilde{s}}(\beta_{n})\ket{\Psi}&=e^{iu_{Q}}\ket{\Psi}\\
\implies\left[\rho_{\tilde{s}}(\beta_{n})\right]^{\dagger}\ket{\Psi}&=e^{-iu_{Q}}\ket{\Psi} \end{align*}
where the second line is well-defined since it can be shown that $s\in\mathcal{U}_{\tilde{s},\beta_{n}}$.
\end{proof}

\vspace{4mm}
\noindent In order to prove Theorem \ref{main1}, we will need the braid identity in Lemma \ref{main1lem3} (whose proof is given in Appendix \ref{proofappx1}).\\

\begin{lem}Let $n\geq2$. Then for  $i=1,\ldots,n-1$,
\begin{equation}\beta_{n}\sigma^{\pm1}_{i}=\sigma^{\pm1}_{n-i}\beta_{n} \end{equation}\label{main1lem3}\end{lem}

\begin{proof}[\underline{Proof of Theorem \ref{main1}}]\hspace{2mm}\\[2mm]
Take $n$-quasiparticle space $V^{s}$ for some chosen $s\in S_{\{1,\ldots,n\}}$. We write the eigenspace decomposition $V^{s}=\bigoplus_{Q}V^{s}_{Q}$ under $\rho_{s}(\beta_{n})$ where
\begin{equation}\begin{split}\rho_{s}(\beta_{n}):V^{s}_{Q}&\to V^{\beta_{n}(s)}_{Q} \\
\ket{\Psi}&\mapsto e^{iu_{Q}}\ket{\Psi}\end{split}\ , \ n\geq2\end{equation}
Then for $1\leq i\leq n-1$,
\begin{equation*}\rho_{s}(\beta_{n}\sigma^{\pm1}_{i})\ket{\Psi}=\rho_{\sigma_{i}(s)}(\beta_{n})\left[\rho_{s}(\sigma^{\pm1}_{i})\ket{\Psi}\right]\end{equation*}
and
\begin{alignat*}{1}\rho_{s}(\beta_{n}\sigma^{\pm1}_{i})\ket{\Psi}&=\rho_{s}(\sigma^{\pm1}_{n-i}\beta_{n})\ket{\Psi} \quad \text{(by Lemma \ref{main1lem3})}\\
&=e^{iu_{Q}}\left[\rho_{\beta_{n}(s)}(\sigma^{\pm1}_{n-i})\ket{\Psi}\right] \end{alignat*}
where $\sigma_{i}(s)$ swaps the $i^{th}$ and $(i+1)^{th}$ characters of $s$, and $\beta_{n}(s)$ reverses the order of the characters in $s$. Then by (\ref{transhit}), we have \[\rho_{\beta_{n}(s)}(\sigma^{\pm1}_{n-i})\ket{\Psi}=\rho_{s}(\sigma^{\pm1}_{i})\ket{\Psi}\] 
It follows that the image of $V^{s}_{Q}$ under $\rho_{s}(\sigma^{\pm1}_{i})$ is the $e^{iu_{Q}}$-eigenspace of $\rho_{\sigma_{i}(s)}(\beta_{n})$, so we write 
\[ \rho_{s}(\sigma^{\pm1}_{i})\left(V^{s}_{Q}\right)=:V^{\sigma_{i}(s)}_{Q} \] 
The result follows.
\end{proof}\vspace{3mm}

\subsection{Fusion structure}
\label{quaf}
A composite collection of quasiparticles will exhibit the same statistical behaviour as a single quasiparticle under exchanges: the scheme under which a collection of quasiparticles is considered as a composite is known as \textit{fusion}. In this section, we will carefully show the emergence of this behaviour by considering the action of the superselection braid.


\vspace{4mm}
\begin{defn}We define $t_{k,l}$ to be the braid in $B_{k+l}$ that clockwise exchanges $k$ strands with $l$ strands. Similarly, we define $u_{k,l}$ to be the braid in $B_{k+l}$ that anticlockwise exchanges $k$ strands with $l$ strands. Clearly,  $t_{k,l}^{-1}=u_{l,k}$. 
\begin{figure}[H]\centering \includegraphics[width=0.55\textwidth]{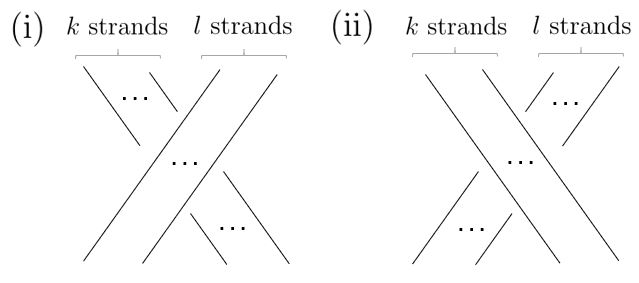} \caption{(i) $t_{k,l}$\ , \ (ii) $u_{k,l}$} \label{t_k,l_braids}\end{figure}\end{defn}

\vspace{4mm}
\noindent For any $a\in\mathbb{N}_{0}$, we have the homomorphism
\begin{equation}\begin{split}r_{a}:B_{n}&\to B_{n+a} \\ \sigma_{i} &\mapsto \sigma_{i+a} \end{split}\end{equation}
where $r_{a_{1}}\circ r_{a_{2}}=r_{a_{1}+a_{2}}$. We also have the anti-automorphism 
\begin{equation}\begin{split}\chi:B_{n}&\to B_{n} \\ \sigma_{i}&\mapsto\sigma_{i}\end{split}\end{equation}
which reverses the order of the generators in a braid word. Let $\overleftarrow{b}:=\chi(b)$. Note that
\begin{equation}\begin{split}t_{k,l}&=r_{0}(\overleftarrow{b_{l}})\cdot r_{1}(\overleftarrow{b_{l}})\cdot\ldots\cdot r_{k-1}(\overleftarrow{b_{l}}) \\
&=r_{l-1}(b_{k})\cdot\ldots\cdot r_{1}(b_{k})\cdot r_{0}(b_{k})\end{split}\label{tikell}\end{equation}\vspace{1mm}
\noindent and that $\overleftarrow{t_{k,l}}=t_{l,k}$. 

\vspace{10mm}
\noindent Consider some $n$-quasiparticle system $V^{s}_{Q}$ in fixed superselection sector $Q$ for some \mbox{$s\in S_{\{1,\ldots,n\}}$ where $n\geq2$.} Partition $s$ into nonempty substrings $m_{1},m_{2}$ i.e. \mbox{$V^{s}_{Q}=V^{m_{1},m_{2}}_{Q}$} and denote the length of string $m_{i}$ by $|m_{i}|$. We write eigenspace \mbox{decompositions} 
\begin{equation} V^{m_{1}}=\bigoplus_{X}V^{m_{1}}_{X} \quad , \quad V^{m_{2}}=\bigoplus_{Y}V^{m_{2}}_{Y}\end{equation}
under $\rho_{m_{1}}(\beta_{|m_{1}|})$ and $\rho_{m_{2}}(\beta_{|m_{2}|})$. Similarly to (\ref{wineup_intro}), we have the decompositions
\begin{subequations}
\begin{align}V^{m_{1},m_{2}}_{Q}&\cong\bigoplus_{X,Y}V^{m_{1}}_{X}\otimes V^{XY}_{Q}\otimes V^{m_{2}}_{Y}\label{ausevy}\\
V^{m_{2},m_{1}}_{Q}&\cong\bigoplus_{X,Y}V^{m_{2}}_{Y}\otimes V^{YX}_{Q}\otimes V^{m_{1}}_{X}\label{ausevy2}
\end{align}\end{subequations}

\newpage
\begin{thm}{\normalfont\textbf{(Fusion)}}\\ 
For an $n$-quasiparticle system $V^{s}_{Q}$ with fixed superselection sector $Q$, consider its decomposition as in (\ref{ausevy}).
Let $(k,l):=(|m_{1}|,|m_{2}|)$ and take $(X,Y)=(x,y)$ such that $V^{xy}_{Q}$ is nonzero. Take arbitrary $\ket{\psi}:=\ket{\psi_{x}}\ket{\psi^{xy}_{Q}}\ket{\psi_{y}}\in V^{m_{1}}_{x}\otimes V^{xy}_{Q}\otimes V^{m_{2}}_{y}$ where we have eigenvalues 
\[\rho_{m_{1}}(\beta_{k})\ket{\psi_{x}}=e^{iu_{x}}\ket{\psi_{x}} \ , \ \rho_{m_{2}}(\beta_{l})\ket{\psi_{y}}=e^{iu_{y}}\ket{\psi_{y}} \ , \ \rho_{s}(\beta_{k+l})\ket{\psi}=e^{iu_{Q}}\ket{\psi} \]
Then,
\vspace{1.25mm}
\begin{enumerate}[label=(\roman*)]
\item $\rho_{s}(t_{k,l})\ket{\psi}=e^{i(u_{Q}-u_{x}-u_{y})}\ket{\psi}$
\vspace{2mm}
\item Eigenspaces are preserved under exchanges i.e.
\begin{equation}\rho_{s}(t_{k,l}):V^{m_{1}}_{x}\otimes V^{xy}_{Q}\otimes V^{m_{2}}_{y}\stackrel{\sim}{\to} V^{m_{2}}_{y}\otimes V^{yx}_{Q}\otimes V^{m_{1}}_{x}\label{yumbo}\end{equation} 
\label{exthm1}
\item $\rho_{m_{2},m_{1}}(t_{l,k})\left[\rho_{m_{1},m_{2}}(t_{k,l})\ket{\psi}\right]=e^{i(u_{Q}-u_{x}-u_{y})}\left[\rho_{m_{1},m_{2}}(t_{k,l})\ket{\psi}\right]$, and so
\begin{equation}\rho_{s}(t_{l,k}\cdot t_{k,l})\ket{\psi}=e^{i2(u_{Q}-u_{x}-u_{y})}\ket{\psi}\end{equation}
\end{enumerate}
\label{mainthm2}\end{thm}


\vspace{4mm}
\begin{cor}
\begin{equation}\rho_{s}(u_{k,l})\ket{\psi}=e^{-i(u_{Q}-u_{x}-u_{y})}\ket{\psi}\end{equation}
\label{negtkl}\end{cor}
\begin{proof}
\begin{alignat*}{4}
&\left[\rho_{m_{2},m_{1}}(t_{l,k})\right]^{\dagger}\rho_{m_{2},m_{1}}(t_{l,k})\rho_{m_{1},m_{2}}(t_{k,l})\ket{\psi}=\rho_{m_{1},m_{2}}(t_{k,l})\ket{\psi} \\
{\implies}& \rho_{m_{1},m_{2}}(u_{k,l})\left[e^{i2(u_{Q}-u_{x}-u_{y})}\ket{\psi}\right]=e^{i(u_{Q}-u_{x}-u_{y})}\ket{\psi}
\end{alignat*}
\end{proof}
\vspace{4mm}
\noindent Theorem \ref{mainthm2} tells us that the $k$ and $l$-quasiparticle composites $m_{1}$ and $m_{2}$ (in eigenstates of $\rho_{m_{1}}(\beta_{k})$ and $\rho_{m_{2}}(\beta_{l})$ respectively) behave identically to a pair of quasi-particles under exchange: if we fix eigenspaces $V^{m_{1}}_{x}$ and $V^{m_{2}}_{y}$ such that $V^{xy}_{Q}$ is nonzero, then composites $m_{1}$ and $m_{2}$ behave as a pair of quasiparticles in superselection \mbox{sector} $Q$ with exchange phase $e^{i(u_{Q}-u_{x}-u_{y})}$. The eigenspaces of $\rho_{m_{1}}(\beta_{k})$ and $\rho_{m_{2}}(\beta_{l})$ may thus be considered as representing different `types' of \mbox{quasiparticles} (since the exchange phase depends on $x$ and $y$). We will refer to the `type' of a quasiparticle as its (topological) \textit{charge}. If e.g. $k>1$, we say that the collection $m_{1}$ of quasiparticles \textit{fuses} to a quasiparticle of charge $x$.\footnote{ For $k$=1, note that $u_{x}=0$ since the eigenvalue of $\rho_{m_{1}}(\beta_{1})$ is trivial. Let $m_{1}=q_{j}$. In (\ref{yumbo}), we write $x=q_{j}$ i.e. $q_{j}$ `fuses to itself'. Note that $V^{q_{j}}_{q_{j}}=V^{q_{j}}$ since the eigenspace is the whole space, and recall that a $1$-quasiparticle Hilbert space is canonically isomorphic to $\mathbb{C}$.} It follows that the possible $(x,y)$ for which $V^{xy}_{Q}$ is nonzero represent the distinct possible \textit{fusion outcomes} here. \\ \\
Recall from Remark \ref{supersubdestruct} that we can have a coherent superposition of distinct fusion outcomes for an entangled subsystem of quasiparticles. Furthermore, since the eigenspaces of any $\rho_{\Sigma}(\beta_{n})$ (where $\Sigma$ is an unordered set of quasiparticles of cardinality $n$) can be identified with quasiparticle charges, it follows that the superselection sector of a \mbox{system} can be identified with a (composite) quasiparticle of fixed charge. A complete system of quasiparticles thus has fixed total charge (fusion outcome).

\begin{figure}[H]\centering \includegraphics[width=0.65\textwidth]{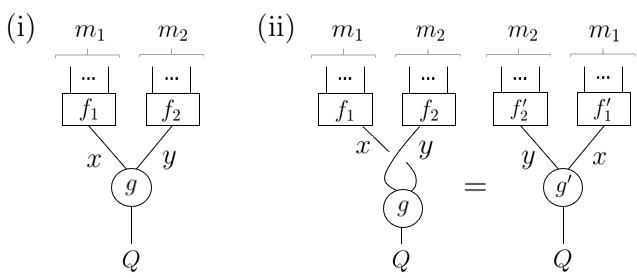} \caption{(i) The fusion diagram graphically depicting an arbitrary state in \mbox{$V^{m_{1}}_{x}\otimes V^{xy}_{Q}\otimes V^{m_{2}}_{y}$} where $f_{1}\in V^{m_{1}}_{x}$, $f_{2}\in V^{m_{2}}_{y}$ and $g\in V^{xy}_{Q}$. \mbox{(ii) Composite} charges $x$ and $y$ are exchanged in superselection sector $Q$, so the fusion state acquires phase $e^{i(u_{Q}-u_{x}-u_{y})}$ relative \mbox{to (i).}} \label{fusdiag1}\end{figure}

\noindent This lends the hitherto abstract factor $V^{xy}_{Q}$ in (\ref{yumbo}) a more tangible interpretation: $V^{m_{1}}_{x}\otimes V^{xy}_{Q}\otimes V^{m_{2}}_{y}$ is the space of states describing the process where collection $m_{1}$ fuses to (a quasiparticle of charge) $x$, collection $m_{2}$ fuses to $y$, and then $x$ and $y$ fuse to $Q$ (see Figure \ref{fusdiag1}(i)). The interpretation of any such tensor decomposition follows analogously. Such Hilbert spaces are thus known as \textit{fusion spaces} and their constituent states are called \textit{fusion states}.\vspace{6mm}

\begin{cor}Fusion is commutative and associative.\label{fucoma}\end{cor}
\begin{proof}Commutativity follows from Theorem \ref{main1}: the possible fusion outcomes for an $n$-quasiparticle system correspond to the eigenspaces of $\rho_{[n]}(\beta_{n})$ on $\mathcal{V}^{[n]}$ (whence the order of the $n$ quasiparticles is irrelevant). \\
Associativity follows from recursive application of Theorem \ref{mainthm2} i.e.\ further partitioning $m_{1}$ and $m_{2}$ and so on until no further partitions can be made: we will view such a recursive choice of partitions as a \textit{full rooted binary tree with $n$ leaves}. This provides us with a \textit{fusion tree} illustrating the order in which $n$ quasiparticles are fused (see Figure \ref{fustree1}). Since $Q$ corresponds to an arbitrary eigenspace of $\rho_{s}(\beta_{n})$, it follows that the set of possible fusion outcomes (i.e.\ the set of possible labels for the root) does not depend on the order in which fusion occurs.\end{proof}
\vspace{6mm}
\begin{figure}[H]\centering \includegraphics[width=0.85\textwidth]{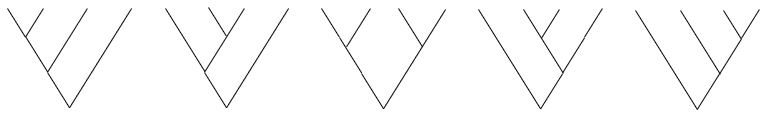} \caption{All possible fusion trees for $4$ particles. For $n$ particles, the number of possible fusion trees is given by $C_{n-1}=\frac{1}{n}{2n-2\choose n-1}$ i.e. the $(n-1)^{th}$ Catalan number.} \label{fustree1} \end{figure}

\newpage
\noindent By the associativity and commutativity of fusion, the charge of an unordered \mbox{collection} $\Sigma$ of quasiparticles can be thought of as a property of any connected region of the system in which solely the excitations in $\Sigma$ are enclosed. This is one of the reasons that quasiparticle charge is called `topological' (as opposed to e.g.\ electric charge which is defined geometrically via the charge density). Similarly to electric charge, we have seen that topological charge may correspond to a superselection rule of a system; but unlike electric charge, we may also observe a superposition of topological charges (for an entangled subsystem).

\vspace{3mm}
\begin{figure}[H]\centering \includegraphics[width=0.3\textwidth]{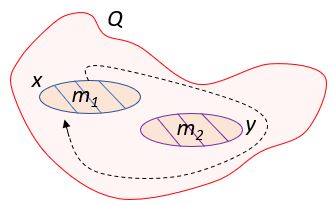}\caption{Winding a quasiparticle collection $m_{1}$ of charge $x$ around collection $m_{2}$ of charge $y$ in a region of total charge $Q$ accumulates statistical phase $e^{i2(u_{Q}-u_{x}-u_{y})}$. This diagram illustrates the same process as on the left-hand side of Figure \ref{fusdiag1}(ii) but with an additional exchange.}\label{windcolls}\end{figure}

\vspace{1mm}
\begin{remk}Take care to note that statistical phases of the form $e^{iu_{Q}}$ are not a property of charge $Q$ alone, but arise as eigenvalues of some $\rho_{s}(\beta_{n})$ i.e.\ the phase also depends on the constituent charges fusing to $Q$. To this end, a better notation for $e^{iu_{Q}}$ would be $\omega^{\Sigma}_{Q}\in U(1)$ where $\Sigma$ is the unordered set of constituent characters of $s$. Nonetheless, we have opted for the former notation for sake of presentation.\label{ssphaseinpdep}\end{remk}

\noindent As indicated by Theorem \ref{mainthm2}, fusion generally does not correspond to a physical process but rather describes how a collection of charges may be considered as a composite charge. Of course, the \textit{measurement} of a fusion outcome is \mbox{physically significant.}\\


\noindent In order to prove Theorems \ref{mainthm2}, we will need the braid identities in Theorem \ref{ssbrecur} (whose proof is given in Appendix \ref{proofappx2}). Theorem \ref{ssbrecur} shows that the superselection braid may be defined recursively.\footnote{Choosing between forms (i)-(iv) at each decision (and permuting the terms in square brackets if desired) parses $\beta_{n}$ into a composition of braids of the form $r_{d}(t_{k,l})$.\ The braid word (\ref{ssbprop}) for $\beta_{n}$ is recovered by choosing (ii) at every iteration with $l=1$.}\textsuperscript{,}\footnote{Note that $\beta_{n}^{-1}$ is given by (i)-(iv) but with a superscript `$-1$' on each $t$ and $\beta$. This is easily seen by inverting (i)-(iv).}

\vspace{3mm}
\begin{thm}{\normalfont\textbf{(Superselection braid by recursion)}}\\ 
Let $n\geq2$. For any positive integers $k, l$ such that $k+l=n$, $\beta_{n}$ is given by
\begin{enumerate}[label=(\roman*)]
\item $\left[\beta_{l}\cdot r_{l}(\beta_{k})\right]t_{k,l}$
\item $t_{k,l}\left[\beta_{k}\cdot r_{k}(\beta_{l})\right]$
\item $\beta_{l}\cdot t_{k,l}\cdot\beta_{k}$
\item $r_{l}(\beta_{k})\cdot t_{k,l}\cdot r_{k}(\beta_{l})$
\end{enumerate}
and $\beta_{1}:=e$. The terms enclosed in square brackets commute.
\label{ssbrecur}\end{thm}

\begin{proof}[\underline{Proof of Theorem \ref{mainthm2}}]\hspace{2mm}\\[2mm]
Let $\tilde{v}$ denote the reverse of a string $v$.

\vspace{1.5mm}
\begin{enumerate}[label=(\roman*)]
\item Using Theorem \ref{ssbrecur}(ii),
\begin{equation*}\begin{split}
\rho_{s}(\beta_{n})\ket{\psi}&=\rho_{\tilde{m}_{1},\tilde{m}_{2}}(t_{k,l}) \ \rho_{m_{1},m_{2}}\left([\beta_{k}\cdot r_{k}(\beta_{l})]\right)\ket{\psi} \\
&=\rho_{\tilde{m}_{1},\tilde{m}_{2}}(t_{k,l})\left[e^{i(u_{x}+u_{y})}\ket{\psi}\right]
\end{split}\end{equation*}
Recalling that $\rho_{s}(\beta_{n})\ket{\psi}=e^{iu_{Q}}\ket{\psi}$, we deduce that
\begin{equation*}\begin{split}
\rho_{\tilde{m}_{1},\tilde{m}_{2}}(t_{k,l}):V^{\tilde{m}_{1}}_{x}\otimes V^{xy}_{Q}\otimes V^{\tilde{m}_{2}}_{y}&\to V^{\tilde{s}}_{Q} \\
\ket{\phi}&\mapsto e^{i(u_{Q}-u_{x}-u_{y})}\ket{\phi}
\end{split}\end{equation*}

\vspace{1mm}
\item We know that
\begin{equation}\rho_{s}(t_{k,l}): V^{m_{1}}_{x}\otimes V^{xy}_{Q}\otimes V^{m_{2}}_{y} \to V^{m_{2}m_{1}}_{Q}\label{famalamuto}\end{equation}
where $V^{m_{2},m_{1}}_{Q}$ has decomposition (\ref{ausevy2}). We wish to show that the range of (\ref{famalamuto}) is restricted as in (\ref{yumbo}). Using Theorem \ref{mainthm2}(i) and Theorem \ref{ssbrecur}(iii),
\begin{equation*}\begin{split}
\rho_{s}(\beta_{n})\ket{\psi}&=\rho_{m_{2},\tilde{m}_{1}}(\beta_{l})\ \rho_{\tilde{m}_{1},m_{2}}(t_{k,l})\ \rho_{m_{1},m_{2}}(\beta_{k})\ket{\psi} \\
&=\rho_{m_{2},\tilde{m}_{1}}(\beta_{l})\left[e^{i(u_{Q}-u_{y})}\ket{\psi}\right]
\end{split}\end{equation*}
and since $\rho_{s}(\beta_{n})\ket{\psi}=e^{iu_{Q}}\ket{\psi}$, we deduce that
\begin{equation} \rho_{\tilde{m}_{1},m_{2}}(t_{k,l}): V^{\tilde{m}_{1}}\otimes V^{xy}_{Q}\otimes V^{m_{2}}_{y}\to \bigoplus_{X}V^{m_{2}}_{y}\otimes V^{yX}_{Q}\otimes V^{\tilde{m}_{1}}_{X} \label{gloopy1}\end{equation}
Similarly, by using Theorem  \ref{mainthm2}(i) and Theorem \ref{ssbrecur}(iv) we may deduce that
\begin{equation} \rho_{m_{1},\tilde{m}_{2}}(t_{k,l}): V^{m_{1}}\otimes V^{xy}_{Q}\otimes V^{\tilde{m}_{2}}_{y}\to \bigoplus_{Y}V^{\tilde{m}_{2}}_{Y}\otimes V^{Yx}_{Q}\otimes V^{m_{1}}_{x} \label{gloopy2}\end{equation}
Combining (\ref{gloopy1}) and (\ref{gloopy2}), the result follows.

\vspace{4mm}
\item By identities (i) and (ii) of Theorem \ref{ssbrecur},
\begin{equation}\beta_{n}^{2}=t_{l,k}\left[r_{l}(\beta_{k}^{2})\cdot\beta^{2}_{l}\right] t_{k,l}\end{equation}
whence
\begin{equation*}\begin{split}\rho_{s}(\beta_{n}^{2})\ket{\psi}&=e^{i2(u_{x}+u_{y})}\left[\rho_{m_{2},m_{1}}(t_{l,k})\cdot \rho_{m_{1},m_{2}}(t_{k,l})\ket{\psi}\right] \\ 
\implies \rho_{m_{2},m_{1}}(t_{l,k})\left[\rho_{m_{1},m_{2}}(t_{k,l})\ket{\psi}\right]&=e^{i2(u_{Q}-u_{x}-u_{y})}\ket{\psi} \\
\implies \rho_{m_{2},m_{1}}(t_{l,k})\left[\rho_{m_{1},m_{2}}(t_{k,l})\ket{\psi}\right]&=e^{i(u_{Q}-u_{x}-u_{y})}\left[\rho_{m_{1},m_{2}}(t_{k,l})\ket{\psi}\right]\end{split}\end{equation*}
where we used parts (i) and (ii) of Theorem \ref{mainthm2} in the third and first lines respectively.
\end{enumerate}
\end{proof}

Given the fusion space $V^{s}=\bigoplus_{Q}V^{s}_{Q}$ (where $s=q_{1}\ldots q_{n}\in S_{\{1,\ldots,n\}}$ and $Q$ indexes the superselection sectors), fix a fusion tree (as in Figure \ref{fustree1}): each of the $n-1$ fusion vertices\footnote{By ``fusion vertices'', we mean vertices in the fusion tree with two or more incident edges i.e.\ any vertex that is not a leaf. Leaves correspond to initial quasiparticles.} corresponds to an eigenspace of $\rho_{s(v)}(\beta_{|s(v)|})$, where for a fusion vertex $v$ we let $s(v)$ denote the substring of $s$ given by the leaves descending from $v$, and $|s(v)|$ denotes the length of $s(v)$. Note that $2\leq |s(v)| \leq n$. \\

We thus label each fusion vertex $v$ with an eigenspace of $\rho_{s(v)}(\beta_{|s(v)|})$ (recall that such a label represents a fixed topological charge and is called a `fusion outcome' in this context). Such a labelling is called \textit{admissible} if the corresponding fusion subspace of $V^{s}$ has nonzero dimension. Note that the root label corresponds to the superselection sector of the system. Observe that fixing a fusion tree specifies a decomposition of $V^{s}$ in terms of the eigenspaces of $\{\rho_{s(v)}(\beta_{|s(v)|})\}_{v}$ . We write such a decomposition in the form yielded by recursive application 
of (\ref{ausevy}) e.g.\ a fusion tree of the form illustrated in Figure \ref{fustree2} specifies the decomposition \begin{equation}V^{q_{1}q_{2}q_{3}q_{4}}\cong\bigoplus_{X_{1},X_{2},Q}V^{q_{1}q_{2}}_{X_{1}}\otimes V^{X_{1},q_{3}}_{X_{2}}\otimes V^{X_{2},q_{4}}_{Q}\end{equation}
\vspace{-2mm}
\begin{figure}[H]\centering \includegraphics[width=0.2\textwidth]{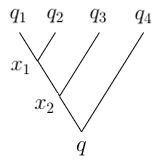} \caption{The labels $x_{1}, x_{2}$ and $q$ correspond to eigenspaces of $\rho_{q_{1}q_{2}}(\beta_{2}), \rho_{q_{1}q_{2}q_{3}}(\beta_{3})$ and $\rho_{q_{1}q_{2}q_{3}q_{4}}(\beta_{4})$ respectively. The triple $(x_{1},x_{2},q)$ of charges is an admissible labelling of the tree if the fusion subspace $V^{q_{1}q_{2}}_{x_{1}}\otimes V^{x,q_{3}}_{x_{2}}\otimes V^{x_{2},q_{4}}_{Q}\subseteq V^{q_{1}q_{2}q_{3}q_{4}}$ is nonzero.} \label{fustree2}\end{figure}

Theorem \ref{ssbrecur} provides a method for parsing $\beta_{n}$ into a composition of braids of the form $r_{d}(t_{k,l})$. Any such parsing involves making a choice of $n-1$ partitions. From any possible sequence of partitions, we can always recover a fusion tree with which the parsing of $\beta_{n}$ is \textit{compatible}. By compatibility, we mean that it is readily apparent how the fusion tree will transform under the action of $\beta_{n}$ i.e.\ $\beta_{n}$ can be parsed into a sequence of braids that each have a well-defined action on the decomposed components of the system. The incoming branches of each fusion vertex in the tree are clockwise exchanged and so the initial fusion tree is sent to its mirror image. The braid $\beta_{n}$ is thus compatible with all $n$-leaf fusion trees (as expected).

\begin{figure}[H]\centering \includegraphics[width=0.4\textwidth]{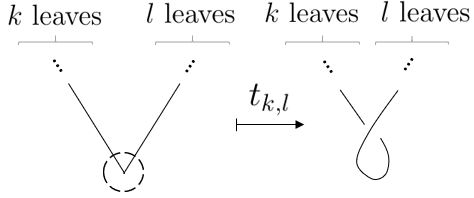} \caption{$t_{k,l}$ clockwise exchanges the incoming branches of a fusion vertex that has $k$ leaves and $l$ leaves stemming from it.} \label{t_k,l}\end{figure}

\begin{remk}Given $\ket{\psi}\in V^{s}_{Q}$ , we know that $\rho_{s}(\beta_{n})\ket{\psi}=e^{iu_{Q}}\ket{\psi}$. It is illuminating to examine how the phase $e^{iu_{Q}}$ arises given a decomposition of $V^{s}_{Q}$ . Consider any admissibly labelled fusion tree in  $V^{q_{1}\ldots q_{n}}_{Q}$ (whence the root has label $Q$). We know that $\rho_{s}(\beta_{n})$ will clockwise exchange the incoming branches of every fusion vertex. For any fusion vertex, the clockwise exchange is given by 
\begin{figure}[H]\centering \includegraphics[width=0.55\textwidth]{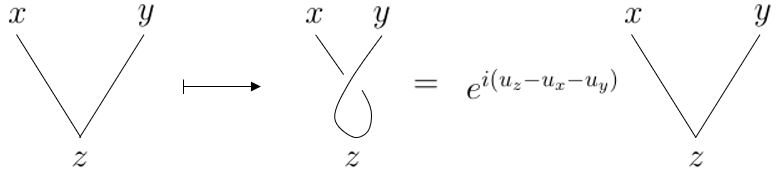} \caption*{} \label{fusvextkl}\end{figure}\vspace{-11mm}
\noindent where the phase evolution follows from Theorem \ref{mainthm2}. It is easy to see that the total phase evolution acquired by clockwise exchanging the incoming branches of every fusion vertex will be $e^{i\left[u_{Q}-(u_{q_{1}}+ \cdots + u_{q_{n}})\right]}$ (phases associated to internal nodes of the tree will cancel). Finally, observe that the $u_{q_{i}}$ are zeroes (since they are arguments of eigenvalues under the action of $\beta_{1}=e$).\label{mainremprec}\end{remk}


\begin{thm}{\normalfont\textbf{(Uniqueness of the superselection braid)}}\\ 
$\beta^{\pm1}_{n}$ are the unique braids under whose action the fusion space decomposes into the superselection sectors of an \mbox{$n$-quasiparticle} system.
\label{sidethm1}\end{thm}
\noindent A proof of Theorem \ref{sidethm1} is outlined in Appendix \ref{ssbfbs}.

\section{Theories of Anyons}
\label{anyonsec}
This section primarily serves to connect the narrative of Section \ref{superbraidsec} with the standard formalism in the literature, by outlining the additional postulates (\textbf{A2}-\textbf{A3}) required to make contact with the usual algebraic theory of anyons. Our presentation thus omits various details, and is not intended as an introduction. For a more detailed treatment, we refer the reader to \cite{kitaev,bonderson,simonox,preskill,wangbook}. In relation to additional insights arising from consideration of the superselection braid, we highlight Section \ref{ssbremk1of2}. 

\subsection{Labels and finiteness}
\label{finduel}

In any standard theory of anyons, it is assumed that there are finitely many distinct topological charges. A theory of anyons thus comes equipped with a finite set of labels $\mathfrak{L}$ whose cardinality is called the \textit{rank} of the theory. 
It is also assumed that the representation space in (\ref{subfixie}) is finite which immediately tells us that $\dim(\mathcal{V}^{\{a,b\}}_{c})$ is finite for any $a,b,c\in\mathfrak{L}$ (from which it easily follows that a fusion space for finitely many quasiparticles is finite-dimensional). We package these two assumptions into the finiteness assumption \textbf{A2} below.

\begin{defn}Given fusion space $V^{ab}_{c}$ for any $a,b,c\in\mathfrak{L}$ , we write $N^{ab}_{c}:=\dim(V^{ab}_{c})$. The quantities $\{N^{ab}_{c}\}_{a,b,c\in\mathfrak{L}}$ are called the \textit{fusion coefficients} of the theory.\end{defn}

\noindent Since $\dim(\mathcal{V}^{\{a,b\}}_{c})=\dim(V^{ab}_{c})=\dim(V^{ba}_{c})$ we have the symmetry 
\begin{equation}N^{ab}_{c}=N^{ba}_{c} \quad \text{for all } a,b,c\in\mathfrak{L}\label{upperfussym}\end{equation}
which is consistent with the commutativity of fusion from Corollary \ref{fucoma}. The quantity $N^{ab}_{c}$ may be thought of as counting `the distinct number of ways charges $a$ and $b$ can fuse to charge $c$'. Note that $\dim(V^{ab})=\sum_{c\in\mathfrak{L}}N^{ab}_{c}$ and that if $N^{ab}_{c}=0$ then $a$ and $b$ cannot fuse to $c$. Consider $V^{abc}_{d}$ for any $a,b,c,d\in\mathfrak{L}$. By associativity of fusion (Corollary \ref{fucoma}), the decompositions of a fusion space must be isomorphic 
\begin{equation}\bigoplus_{e}V^{ab}_{e}\otimes V^{ec}_{d}\cong \bigoplus_{f}V^{af}_{d}\otimes V^{bc}_{f}\label{associso}\end{equation}
and so the fusion coefficients satisfy the associativity relation
\begin{equation}\sum_{e\in\mathfrak{L}}N^{ab}_{e}N^{ec}_{d}=\sum_{f\in\mathfrak{L}}N^{af}_{d}N^{bc}_{f}\end{equation} \\
\fbox{
\parbox{\textwidth}{ 
\noindent\textbf{A2.} A theory of anyons has finitely many distinct topological charges and all fusion coefficients are finite. }}

\vspace{5mm}
Any label set will include the \textit{trivial} label (denoted by $0$) which represents (the topological charge of) the vacuum: the fusion of any charge with the vacuum yields the original charge i.e.\ $N^{0q}_{r}\propto\delta_{qr}$ for any $q,r\in\mathfrak{L}$. Since we always have the freedom to insert the trivial charge anywhere, we must have
\begin{equation}\dim(V^{ab}_{c})=\dim(V^{a0b}_{c})=\dim(V^{0ab}_{c})=\dim(V^{ab0}_{c})\label{poindex}\end{equation}
Associativity and (\ref{poindex}) tell us that $N^{a0}_{a}N^{ab}_{c}=N^{ab}_{c}N^{0b}_{b}=N^{ab}_{c}$ and so $N^{a0}_{a}=N^{b0}_{b}=1$ for all $a,b\in\mathfrak{L}$. Thus, 
\begin{equation}N^{q0}_{r}=N^{0q}_{r}=\delta_{qr} \quad \text{for any }q,r\in\mathfrak{L}\end{equation}
Following the presentation in \cite{kitaev}, write $V^{a0}_{a}=\text{span}_{\mathbb{C}}\{\ket{\alpha_{a}}\}$ and $V^{0b}_{b}=\text{span}_{\mathbb{C}}\{\ket{\beta_{b}}\}$. The relation between the spaces in (\ref{poindex}) is characterised by trivial isomorphisms
\begin{equation}\begin{split}\alpha_{q}:\mathbb{C}&\to V^{q0}_{q} \\ z&\mapsto z\ket{\alpha_{q}} \end{split} \quad\quad\quad\quad \begin{split}\beta_{q}:\mathbb{C}&\to V^{0q}_{q} \\ z&\mapsto z\ket{\beta_{q}}\end{split}\label{alphetiso}\end{equation}
e.g.\ $V^{ab}_{c}\stackrel{\alpha_{a}}{\xrightarrow{\sim}}V^{a0}_{a}\otimes V^{ab}_{c}$ and $V^{ab}_{c}\stackrel{\beta_{b}}{\xrightarrow{\sim}}V^{ab}_{c}\otimes V^{0b}_{b}$. By associativity we see that $\alpha_{a}$ and $\beta_{b}$ are related (see Remark \ref{cocoremk} and Appendix \ref{penthexapx}). Braiding with the vacuum must be trivial i.e.\ using the same notation as in (\ref{subfixie}),
\begin{equation}\rho_{\{q,0\}}(\sigma_{1}^{\pm1})=1\quad\text{ for all $q\in\mathfrak{L}$}\label{trivbr1}\end{equation} 

\vspace{5mm}
\subsection{Braided $6j$ fusion systems}
We write orthonormal bases 
\begin{equation}V^{ab}_{c}=\text{span}_{\mathbb{C}}\{\ket{ab\to c;\mu}\}_{\mu} \ \ , \ \ V^{ba}_{c}=\text{span}_{\mathbb{C}}\{\ket{ba\to c;\mu}\}_{\mu}\label{fusbrba}\end{equation}
of fusion states given any $a,b,c\in\mathfrak{L}$ where $1\leq\mu\leq N^{ab}_{c}$ for $N^{ab}_{c}\neq0$.

\begin{figure}[H]\centering \includegraphics[width=0.1\textwidth]{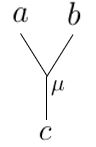} \vspace{-2mm}\caption{A graphical depiction of the fusion state $\ket{ab\to c;\mu}$. We will implicitly assume that our fusion vertices are normalised.} \label{fusstate1}\end{figure}

\noindent The \textit{dual space} of a fusion space has natural interpretation as a `\textit{splitting space}' i.e.\
\vspace{-2mm}
\begin{equation}\begin{gathered}\centering \includegraphics[width=0.7\textwidth]{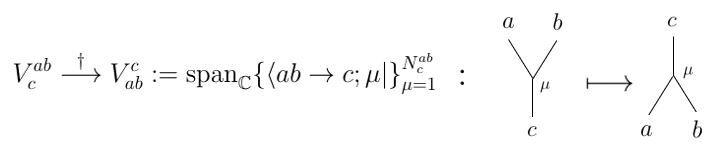}\label{dualconj}\end{gathered}\end{equation}
for any $a,b,c\in\mathfrak{L}$. Fusion coefficients may thus also be thought of `splitting' coefficients. Given an orthonormal basis, we can use the graphical calculus to express the inner product and completeness relation on $V^{ab}$ :
\begin{equation}\begin{gathered}\centering \includegraphics[width=0.65\textwidth]{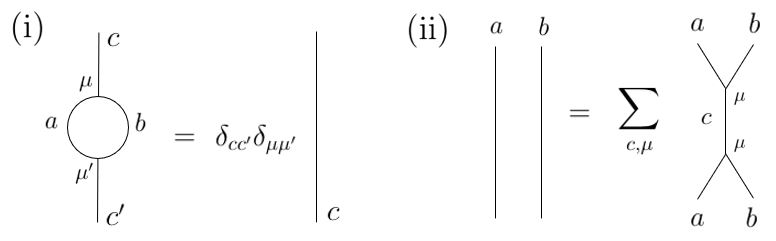}\label{innercomp}\end{gathered}\end{equation}

\vspace{4mm}
\noindent The \textit{$R$-matrices} of a theory are given by a matrix representation of the unitary operators from (\ref{subfixie}), typically in an eigenbasis: given any $a,b\in\mathfrak{L}$ we have the eigenspace decomposition $\mathcal{V}^{\{a,b\}}=\bigoplus_{Q\in\mathfrak{L}}\mathcal{V}^{\{a,b\}}_{Q}$ under $\rho_{\{a,b\}}$ where 
\begin{equation}\rho_{\{a,b\}}(\sigma_{1}^{\pm1})\ket{\psi}=e^{\pm iu_{Q}}\ket{\psi}\end{equation}
for $\ket{\psi}\in\mathcal{V}^{\{a,b\}}_{Q}$ with $Q$ such that $N^{ab}_{Q}\neq0$. We write $R$-matrices
\begin{equation}
R^{ab}_{Q} : V^{ab}_{Q} \stackrel{\sim}{\to} V^{ba}_{Q} \quad , \quad R^{ba}_{Q} : V^{ba}_{Q} \stackrel{\sim}{\to} V^{ab}_{Q}
\end{equation}
where we let
\begin{subequations}\begin{align}R^{ab}_{Q}=R^{ba}_{Q}&=\bigoplus_{j=1}^{N^{ab}_{Q}}[e^{iu_{Q}}]  \label{rdiag1}\\
 R^{ab}:=\bigoplus_{Q\in\mathfrak{L} \ : \ N^{ab}_{Q}\neq0}\left[R^{ab}_{Q}\right] \quad &, \quad R^{ba}:=\bigoplus_{Q\in\mathfrak{L} \ : \ N^{ba}_{Q}\neq0}\left[R^{ba}_{Q}\right]\label{rdiag2}\end{align}\end{subequations}
It is clear that $R^{ab}=R^{ba}$ here.\footnote{R-matrices need not always be diagonal and symmetric in their upper indices. However, our construction has implicitly `fixed a gauge' where this is the case; see Remark \ref{gaugefreedom} and (\ref{rgraphginv}).} Following (\ref{trivbr1}), we have
\begin{equation}R^{q0}_{q}=R^{0q}_{q}=1\label{trivbracu}\end{equation}
for all $q\in\mathfrak{L}$. We let $(R^{-1})^{ab}$ denote the anticlockwise exchange i.e.\
\begin{equation}(R^{ab})^{-1}=(R^{-1})^{ba}\end{equation}

\vspace{6.5mm}
\noindent For an $n$-quasiparticle fusion space $V^{q_{1}\ldots q_{n}}$ (where $q_{1},\ldots,q_{n}\in\mathfrak{L}$) let $\mathscr{D}_{1}$ and $\mathscr{D}_{2}$ be decompositions of this space corresponding to distinct fusion trees. By associativity, we have an isomorphism
\begin{equation}\mathcal{F}:\mathscr{D}_{1}\to\mathscr{D}_{2}\label{calico}\end{equation}
Fixing a basis of fusion states, we see that $\mathcal{F}\in\text{Aut}(V^{q_{1}\ldots q_{n}})$ is a change of basis matrix. Observe that $\mathcal{F}$ is given by any sequence of so-called \textit{F-moves} that transform between decompositions of the form
\begin{figure}[H]\centering \includegraphics[width=0.4\textwidth]{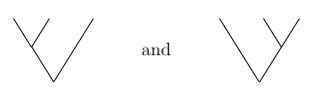} \caption*{} \label{fdecomps}\end{figure}\vspace{-14mm}
\noindent Such transformations are realised by the \textit{F-matrices} of a theory. These are matrices $F^{abc}_{d}\in\text{Aut}(V^{abc}_{d})$ for any $a,b,c,d\in\mathfrak{L}$ where
\begin{equation}F^{abc}_{d}:\bigoplus_{e\in\mathfrak{L}}V^{ab}_{e}\otimes V^{ec}_{d} \xrightarrow{\sim} \bigoplus_{f\in\mathfrak{L}}V^{af}_{d}\otimes V^{bc}_{f}\end{equation}
This is a unitary matrix representing the isomorphism in (\ref{associso}). That is, $F^{abc}_{d}$ transforms between the bases 
\begin{equation}\Big\{\ket{ab\to e;\mu_{1}^{e}}\ket{ec\to d;\mu_{2}^{e}}\Big\}_{e,\mu_{1}^{e},\mu_{2}^{e}} \ \ \text{and} \ \ \ \left\{\ket{af\to d;\nu_{2}^{f}}\ket{bc\to f;\nu_{1}^{f}}\right\}_{f,\nu_{1}^{f},\nu_{2}^{f}}\end{equation}
This change of basis is graphically expressed as
\begin{equation}\centering \includegraphics[width=0.6\textwidth]{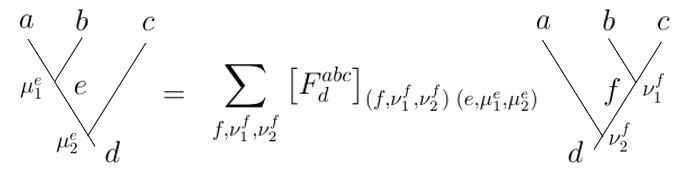} \label{fmatgraph}\end{equation}
\noindent Distinct fusion trees specify distinct bases on the fusion space and are therefore also called \textit{fusion bases}. Since $R^{ab}$ is defined for an eigenbasis of $V^{ab}$, we must fix a fusion basis such that the factors $\{V^{ab}_{Q}\}_{Q\in\mathfrak{L}}$ appear in the decomposition of the fusion space: for any such fusion basis, we say that `$a$ and $b$ are in a \textit{direct fusion channel}'. That is, $R$-matrices can only act on two charges in a direct fusion channel.
\begin{figure}[H]\centering \includegraphics[width=0.3\textwidth]{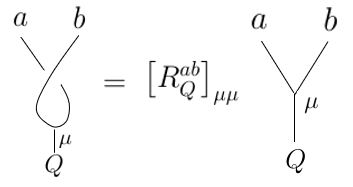} \caption{Charges $a$ and $b$ are in a direct fusion channel with outcome $Q$. The above is a graphical expression of the equation \mbox{$R^{ab}\ket{ab\to Q;\mu}=\left[R^{ab}_{Q}\right]_{\mu\mu}\ket{ab\to Q;\mu}\in\text{span}_{\mathbb{C}}\{\ket{ba\to Q;\mu}\}\subseteq V^{ba}_{Q}$} where the matrix $R^{ab}$ is defined as in (\ref{rdiag1}) and (\ref{rdiag2}).} \label{rchannel}\end{figure}
We may obtain a (possibly non-diagonal) representation of the exchange operator for two adjacent quasiparticles $a$ and $b$ in a system by considering its action with respect to a fusion basis in which $a$ and $b$ are in an \textit{indirect} fusion channel.\footnote{Non-diagonal representations arise since fixing an indirect fusion channel of two charges means that we are not in an eigenbasis of the exchange operator for these charges. Since we are not in an eigenbasis, we cannot apply the R-matrix directly.} Such a representation can be determined by transforming into a fusion basis where the charges are in a \textit{direct} fusion channel, applying the R-matrix and then transforming back to the original fusion basis. Below is the simplest example of such a procedure. \vspace{-7mm}
\begin{figure}[H]\centering \includegraphics[width=0.85\textwidth]{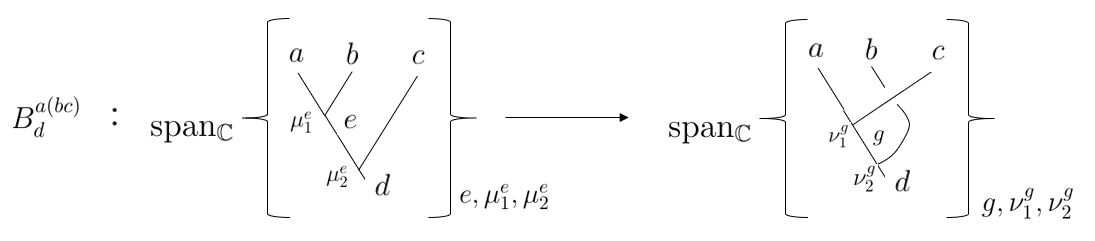} \caption*{} \label{indexchfn}\end{figure}\vspace{-15mm}
\noindent where
\begin{equation}\begin{matrix}
\bigoplus_{e}V^{ab}_{e}\otimes V^{ec}_{d} & \stackrel{\mbox{\normalsize $F^{abc}_{d}$}}{\Larrow{2cm}} & \bigoplus_{f}V^{af}_{d}\otimes V^{bc}_{f} \\
\hspace{-2mm}{B^{a(bc)}_{d}}\Bigg\downarrow &  & \hspace{15mm}\Bigg\downarrow{R^{bc}} \\
\bigoplus_{g}V^{ac}_{g}\otimes V^{gb}_{d} & \stackrel{\mbox{\normalsize $F^{acb}_{d}$}}{\Larrow{2cm}} & \bigoplus_{f}V^{af}_{d}\otimes V^{cb}_{f}
\end{matrix}\end{equation}
That is,
\begin{equation}B^{a(bc)}_{d}=\left(F^{acb}_{d}\right)^{\dagger}R^{bc}F^{abc}_{d}\end{equation}
where 
\begin{equation}R^{bc}=\bigoplus_{f\in\mathfrak{L} \ : \ N^{af}_{d}N^{bc}_{f}\neq0}R^{bc}_{f}\end{equation}

A charge $q\in\mathfrak{L}$ such that $\sum_{u\in\mathfrak{L}}N^{qx}_{u}=1$ for all $x\in\mathfrak{L}$ corresponds to an \textit{abelian} anyon (since its exchange statistics with any other charge will always be given by a phase). Otherwise, $q$ corresponds to a \textit{non-Abelian} anyon (since there exists a charge with which its exchange statistics are given by a higher-dimensional unitary transformation). An abelian theory of anyons is one in which there are no non-abelian anyons.
Observe that given a fixed fusion basis and an explicit choice of orthonormal basis for a fusion space of $n$ \textit{identical charges}, we obtain a unitary matrix representation of the braid group $B_{n}$. 
\begin{remk}\textbf{(Gauge freedom)} \\
There is generally some redundancy amongst the $F$ and $R$ symbols\footnote{$F$ and $R$ symbols refer to the entries of $F$ and $R$ matrices. $F$-symbols are also called \textit{$6j$ symbols}.} of a theory: this arises from the $U(N^{ab}_{c})$ freedom when fixing an orthonormal basis on the spaces $\{V^{ab}_{c}\}_{a,b,c\in\mathfrak{L}}$ . A change of basis\footnote{This is not to be confused with a change of \textit{fusion} basis.} is called a \textit{gauge transformation}. We can only attach physical significance to gauge-invariant quantities. \\
Although $R$-symbols are generally gauge-variant, gauge transformations are defined to respect the triviality of braiding with the vacuum (i.e.\ (\ref{trivbracu}) is gauge-invariant by construction). A \textit{monodromy} is a composition 
\begin{equation}R^{ba}\circ R^{ab}=:M^{ab}\label{monod}\end{equation} 
It can be shown that monodromies are gauge-invariant, whence it follows that the action of any pure braid is gauge-invariant. We implicitly fixed a gauge where \mbox{$R^{ab}=R^{ba}$} for all $a,b\in\mathfrak{L}$ in our construction: we will call this the \textit{symmetric gauge}. R-matrices are not necessarily diagonal and symmetric in their upper indices outside of this gauge. Nonetheless, monodromy matrices are always diagonal and symmetric in their upper indices.
\label{gaugefreedom}\end{remk}

\begin{remk}\textbf{(Coherence conditions)}\\
Isomorphisms between fusion spaces must be `compatible' with one another. That is, distinct sequences of isomorphisms (F-moves, R-moves and isomorphisms $\alpha$ and $\beta$ from (\ref{alphetiso})) between two given spaces should correspond to the same isomorphism. Such compatibility requirements are called \textit{coherence conditions}. Remarkably, all coherence conditions are fulfilled if the triangle, pentagon and hexagon equations are satisfied. Some additional details are provided in Appendix \ref{penthexapx}.

\begin{enumerate}[label=(\roman*)]

\item All isomorphisms $\alpha$ and $\beta$ from (\ref{alphetiso}) must be compatible with associativity (F-moves). This coherence condition is fulfilled if the \textit{triangle equations} (\ref{treqns}) are satisfied.

\item Recall the isomorphism $\mathcal{F}$ from (\ref{calico}). It may be possible that multiple distinct sequences of F-moves realise $\mathcal{F}$. Given some basis, the matrix re-presentation of $\mathcal{F}$ must be the same for all such sequences. This coherence condition is fulfilled if all $F$-symbols satisfy the \textit{pentagon equation} (\ref{pent1}).

\item Consider $n$-quasiparticle space $V^{q_{1}\ldots q_{n}}$ where $q_{1},\ldots,q_{n}\in\mathfrak{L}$ and $n\geq3$. Let $s$ and $s'$ be any two distinct permutations of the string $q_{1}\ldots q_{n}$. Let $\mathcal{D}$ and $\mathcal{D}'$ be any decomposition of $V^{s}$ and $V^{s'}$ respectively. It may be possible that multiple distinct sequences of F and R moves realise the isomorphism $\mathcal{B}:\mathcal{D}\to\mathcal{D}'$ corresponding to the action of some $n$-braid. Given some basis, the matrix representation of $\mathcal{B}$ must be the same for all such sequences. This coherence condition is fulfilled if all $F$ and $R$ symbols satisfy the \textit{hexagon equations} (\ref{hexeqs1}).

\end{enumerate}\label{cocoremk}\end{remk} 
\vspace{2mm}

\noindent For each charge in a theory of anyons, there exists a dual charge wihich with it may fuse to the vacuum. More precisely, we incorporate Kitaev's duality axiom from \cite{kitaev}:\\

\fbox{
\parbox{\textwidth}{
\noindent\textbf{A3.} For each $q\in\mathfrak{L}$, there exists some $\bar{q}\in\mathfrak{L}$ and $\ket{\xi}\in V^{q\bar{q}}_{0}$, $\ket{\eta}\in V^{\bar{q}q}_{0}$ such that
\begin{equation}\bra{\alpha_{q}\otimes\eta}F^{q\bar{q}q}_{q}\ket{\xi\otimes\beta_{q}}\neq0\label{brazlioti}\end{equation}
where $\alpha_{q},\beta_{q}$ are as defined in (\ref{alphetiso}). 
}}

\vspace{2mm}
\begin{prop}\cite[Lemma E.3.]{kitaev} For $q\in\mathfrak{L}$, there exists unique $\bar{q}\in\mathfrak{L}$ such that 
\begin{equation} N^{pq}_{0}=N^{qp}_{0}=\delta_{p\bar{q}}\end{equation} \end{prop}

\noindent This proposition follows from \textbf{A3} and says that any charge has a \textit{unique} dual charge with which it annihilates in a \textit{unique} way. Together with associativity, \textbf{A3} tells us that for any $a,b,c\in\mathfrak{L}$ we have $N^{ab}_{\bar{c}}N^{\bar{c}c}_{0}=N^{a\bar{a}}_{0}N^{bc}_{\bar{a}}$ and so $N^{ab}_{\bar{c}}=N^{bc}_{\bar{a}}$. We thus have 
\begin{equation}N^{ab}_{c}=N^{b\bar{c}}_{\bar{a}}=N^{\bar{c}a}_{\bar{b}}\label{cycsyms}\end{equation}

\vspace{2mm}
\begin{cor}
Any topological charge $q\in\mathfrak{L}$ may realise a superselection sector.
\end{cor}
\renewcommand\qedsymbol{\mbox{\larger\Lightning}}
\begin{proof}
We know that it is possible for a fusion outcome to realise a superselection sector. Suppose there exists a charge $q\in\mathfrak{L}$ such that it is not a fusion outcome for any pair of charges. For any charge $b$ there exists a charge $c$ such that $N^{\bar{q}b}_{c}\neq0$. By (\ref{cycsyms}) we have $N^{\bar{q}b}_{c}=N^{b\bar{c}}_{q}$ which gives a contradiction.
\end{proof}\renewcommand\qedsymbol{$\square$}

\noindent We see that the duality axiom permits any charge to realise a superselection sector. For this reason, labels are often called topological charges and superselection sectors interchangeably in the literature.

\newpage
\noindent For $a,b,c\in\mathfrak{L}$ we define linear maps $K^{ab}_{c}$ and $L^{ab}_{c}$,
\begin{equation}\begin{gathered}\centering \includegraphics[width=0.85\textwidth]{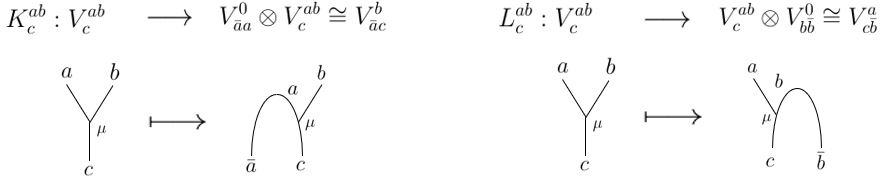}\label{legraises}\end{gathered}\end{equation}
These are clearly invertible (whence $N^{ab}_{c}=N^{\bar{a}c}_{b}=N^{c\bar{b}}_{a}$). Observe that\footnote{The isomorphisms $K^{b\bar{c}}_{\bar{a}}\circ\left(L^{b\bar{c}}_{\bar{a}}\right)^{-1}\circ K^{ab}_{c}$ and $L^{\bar{c}a}_{b}\circ\left(K^{\bar{c}a}_{b}\right)^{-1}\circ L^{ab}_{c}$ correspond to the CPT symmetry of $V^{ab}_{c}$. Indeed, in \cite[Theorem E.6.]{kitaev} it is shown that these two maps are coincide (and are isometries), which is is equivalent to the statement that a unitary fusion category admits a pivotal structure.}
\begin{equation}\begin{gathered}\centering \includegraphics[width=0.5\textwidth]{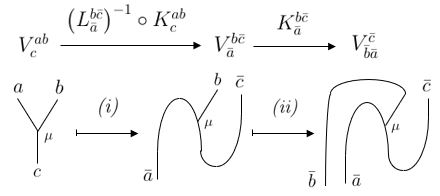}\label{legraisesymms}\end{gathered}\end{equation}
where (i) corresponds to symmetries of the form in (\ref{cycsyms}), and the composition of (i) and (ii) tells us that $N^{ab}_{c}=N^{\bar{b}\bar{a}}_{\bar{c}}$. Together with (\ref{upperfussym}), these identities generate all symmetries of the fusion coefficients. Summarising these, for all $a,b,c\in\mathfrak{L}$ we have 
\begin{subequations}\begin{align}  N^{ab}_{c}&=N^{ba}_{c} \label{syma} \\
 N^{ab}_{c}&=N^{b\bar{c}}_{\bar{a}}=N^{\bar{c}a}_{\bar{b}} \label{symb}\\ 
N^{ab}_{c}&=N^{\bar{b}\bar{a}}_{\bar{c}}\label{symc} \end{align}\end{subequations}  

\vspace{1.5mm}
\begin{defn}Altogether, a finite label set $\mathfrak{L}$ with fusion coefficients, $F$-symbols and $R$-symbols as described above satisfying the triangle, pentagon and hexagon equations is called a \textit{braided $6j$ fusion system}. \end{defn}

\vspace{2mm}
\subsection{Eigenvalues of the superselection braid}
\label{ssbremk1of2}
In Remark \ref{mainremprec}, we examined the action of the superselection braid on any decomposition of the space $V^{s}_{Q}$ (where $s$ is any permutation of some $n$ fixed labels). We know that this action results in the same statistical phase independently of the given permutation or decompositon. Our observations from Remark \ref{mainremprec} look more interesting when recast in terms of R-matrices. Namely, for any choice of labels $1,2,3,4\in\mathfrak{L}$ such that $V^{123}_{4}$ is nonzero, the elements of the table below are equal for any choice of $e,f,g$ such that $N^{12}_{e}N^{e3}_{4}, N^{23}_{f}N^{f1}_{4}$ and $N^{13}_{g}N^{g2}_{4}$ are nonzero and where there exists a choice of gauge such that the relevant $R$-matrices may be written as in (\ref{rdiag1})-(\ref{rdiag2}). 

\renewcommand{\arraystretch}{1.5}
\begin{center}\begin{tabular}{c|c|c|c}
$R^{21}_{e}\otimes R^{e3}_{4}$ & $R^{12}_{e}\otimes R^{e3}_{4}$ & $R^{3e}_{4}\otimes R^{12}_{e}$ & $R^{3e}_{4}\otimes R^{21}_{e}$ \\
\hline
$R^{32}_{f}\otimes R^{f1}_{4}$ & $R^{23}_{f}\otimes R^{f1}_{4}$ & $R^{1f}_{4}\otimes R^{23}_{f}$ & $R^{1f}_{4}\otimes R^{32}_{f}$ \\
\hline
$R^{31}_{g}\otimes R^{g2}_{4}$ & $R^{13}_{g}\otimes R^{g2}_{4}$ & $R^{2g}_{4}\otimes R^{13}_{g}$ & $R^{2g}_{4}\otimes R^{31}_{4}$
\end{tabular}\end{center}\vspace{1mm}
Let $r^{ab}_{c}$ denote the phase $R^{ab}_{c}=r^{ab}_{c}I_{k}$ (where $I_{k}$ is the $k\times k$ identity matrix and $k=N^{ab}_{c}$). Noting that $r^{ab}_{c}=r^{ba}_{c}$ in the fixed gauge, the above equivalences may be expressed as
\begin{equation}r^{12}_{e}r^{e3}_{4}=r^{23}_{f}r^{f1}_{4}=r^{13}_{g}r^{g2}_{4}\label{symssphase}\end{equation}
for any choice of $e,f,g$ as specified above. The identity (\ref{symssphase}) characterises the fact that the statistical evolution under the action of the superselection braid is independent of the fusion basis and order of quasiparticles. However, this identity also has the weakness of being gauge-dependent. We easily obtain a gauge-invariant form of (\ref{symssphase}): writing $M^{ab}_{c}=m^{ab}_{c}I_{k}$ (where $m^{ab}_{c}=m^{ba}_{c}$ is the monodromy phase), 
\begin{equation}m^{12}_{e}m^{e3}_{4}=m^{23}_{f}m^{f1}_{4}=m^{13}_{g}m^{g2}_{4}\label{puressphase}\end{equation}
for any choice of $e,f,g$ as specified above. This gives us the following ansatz: \textit{for every $q\in\mathfrak{L}$ we may assign a quantity $\vartheta_{q}\in U(1)$ such that}
\begin{equation}m^{ab}_{c}=\frac{\vartheta_{c}}{\vartheta_{a}\vartheta_{b}}\quad\text{for all $a,b,c$ such that $N^{ab}_{c}\neq0$}\label{thansatz}\end{equation}
Indeed, this ansatz turns out to be correct (see Section \ref{pivotality}): the quantity $\vartheta_{q}$ is called the topological spin of $q$ and is the phase evolution under a clockwise $2\pi$-rotation of charge $q$. For a system of charges $q_{1},\ldots,q_{n}$ with overall charge $Q$, the gauge-invariant statistical evolution under the action of the pure braid $\beta_{n}^{2}$ is thus given by (\ref{puressphasetwists}) (whose form is consistent with Remark \ref{ssphaseinpdep}).
\begin{equation}\frac{\vartheta_{Q}}{\vartheta_{q_{1}}\cdot\ldots\cdot\vartheta_{q_{n}}}\label{puressphasetwists}\end{equation}

\vspace{2mm}
\subsection{Fusion algebras and their categorification}
\begin{defn}
Let $\mathbb{Z}B$ be a free $\mathbb{Z}$-module with finite basis $B=\{b_{i}\}_{i\in I}$. We equip $\mathbb{Z}B$ with a bilinear product
\begin{equation*}\begin{split}\cdot : \mathbb{Z}B\times\mathbb{Z}B&\to\mathbb{Z}B \\
(b_{i},b_{j})&\mapsto\sum_{k\in I}c^{ij}_{k}b_{k} \ \ , \ c^{ij}_{k}\in\mathbb{N}_{0}
\end{split}\end{equation*}
such that the following hold for all $i,j,k\in I$ :\begin{enumerate}[label=(\roman*)]
\item There exists an element $\mathbbm{1}:=b_{0}\in B$ such that $\mathbbm{1}\cdot b_{i}=b_{i}\cdot\mathbbm{1}=b_{i}$ 
\item $(b_{i}\cdot b_{j})\cdot b_{k}=b_{i}\cdot(b_{j}\cdot b_{k})$ 
\item $\sum_{l\in I}c^{ij}_{l}>0$
\item There exists an involution $i\mapsto i^{*}$ of $I$ such that $c^{ij}_{0}=c^{ji}_{0}=\delta_{i*j}$ 
\end{enumerate}
The unital, associative $\mathbb{Z}$-algebra $\mathcal{A}=(\mathbb{Z}B,\cdot \ )$ satisfying the above is called a \textit{fusion algebra}. If we also have (v) then $\mathcal{A}$ is called a \textit{commutative} fusion algebra.
\begin{enumerate}[resume,label=(\roman*)]
\item $b_{i}\cdot b_{j}=b_{j}\cdot b_{i}$ 
\end{enumerate}
\end{defn}

\vspace{1mm}
\noindent The quantities $\{c^{ij}_{k}\}_{i,j,k\in I}$ act as the structure constants of a fusion algebra. We can also express properties (i),(ii) and (v) in terms of these constants: \mbox{(i) $c^{i0}_{j}=c^{0i}_{j}=\delta_{ij}$,} (ii) $\sum_{p}c^{ij}_{p}c^{pk}_{u}=\sum_{r}c^{ir}_{u}c^{jk}_{r}$ and (v) $c^{ij}_{k}=c^{ji}_{k}$. The structure constants clearly have symmetries of the same form as in (\ref{symb}) (and (\ref{syma}) for a commutative algebra). Observing that the $^{*}$-involution may be extended to an anti-automorphism of $\mathcal{A}$, it easily follows that the structure constants also have symmetry of the form (\ref{symc}). \\ 

\noindent A commutative fusion algebra $\mathcal{A}$ admits a categorification if there exists a braided $6j$ fusion system with label set $\mathfrak{L}$ and a bijection $\phi:B\to\mathfrak{L}$ such that $c^{ij}_{k}=N^{\phi(i)\phi(j)}_{\phi(k)}$ for all $i,j,k\in B$. It is possible for a given $\mathcal{A}$ to admit more than one categorification, although only finitely many\footnote{This result is known as \textit{Ocneanu rigidty}.} (up to gauge equivalence and relabelling). The categorification of $\mathcal{A}$ yields a braided fusion category (whose skeletal data is given by the braided $6j$ fusion system). From a physical perspective, we are only interested in categories for which (there exists a choice of gauge where) all associated $F$ and $R$ symbols are unitary; namely, \textit{unitary} braided fusion categories. \\ 

\subsection{Ribbon structure}
\label{pivotality}
It is known that a unitary braided fusion category admits a unique unitary ribbon structure \cite{eno,galindo}. In terms of the $R$-symbols of the category, this means that for every $q\in\mathfrak{L}$, there exists a quantity $\vartheta_{q}\in U(1)$ such that the ribbon relation (\ref{ribpro}) is fullfilled. This tells us that given a unitary braided $6j$ fusion system, the ansatz (\ref{thansatz}) is correct and has a unique set of solutions. 
\begin{equation}\sum_{\lambda}\left[R^{ba}_{c}\right]_{\mu\lambda}\left[R^{ab}_{c}\right]_{\lambda\nu}=\frac{\vartheta_{c}}{\vartheta_{a}\vartheta_{b}}\delta_{\mu\nu}\label{ribpro}\end{equation}
Physically, $\vartheta_{q}$ is the phase evolution induced by a clockwise $2\pi$-rotation of a charge $q$, and is called its \textit{topological spin}. The topological spins are roots of unity \cite{kitaev,vafa} and are gauge-invariant. The ribbon relation allows us to promote quasiparticle worldlines to worldribbons, or equivalently tells us how to evaluate type-I Reidemeister moves on worldlines (Figure \ref{topospinssconn}).

\begin{figure}[H]\centering \includegraphics[width=0.85\textwidth]{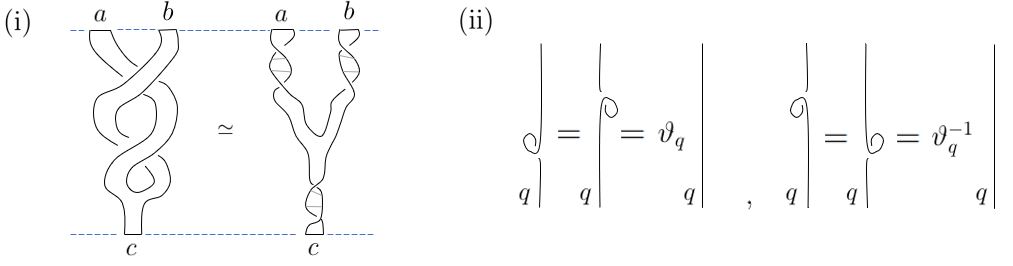} \caption{(i) The ribbon relation illustrated through the deformation of worldribbons. Boundaries are fixed at the initial and final time slices. (ii) Type-I Reidemeister twists correspond to $2\pi$-rotations.} \label{topospinssconn}\end{figure}

\noindent To summarise, the algebraic structure arising from exchange symmetry in two spatial dimensions (under assumptions \textbf{A1}-\textbf{A3}) corresponds to a unitary ribbon fusion category (also called a unitary premodular category). A theory of anyons has all of its data contained in a such a category and is determined (up to gauge equivalence) by the skeletal data of the category \mbox{(fusion} coefficients, $F$-symbols and $R$-symbols). The underlying fusion algebra encodes the \textit{fusion rules} of the theory.\footnote{Note that for the basis $B$ of the fusion algebra for an abelian theory of anyons, $(B,\cdot\ )$ defines an abelian group.} The rank-finiteness theorem for braided fusion categories \cite{rfgxbfc} tells us that there are finitely many theories of anyons for any given rank. Finally, we note that the deduction in Remark \ref{arbtopremk} is verified, for instance, by the toric code modular tensor category which describes quasiparticles on a torus.
 
\subsection{Modularity}
\label{modulasec}

Pursuing a classification of theories of anyons motivates that of unitary ribbon fusion categories \cite{bruillortiz}. Levying a nondegeneracy condition on the braiding results in a unitary \textit{modular tensor category}: the extra structure possessed by such categories makes their classification more tractable \cite{RSW, rk5clas,dgreen}.

\begin{defn}Suppose monodromy operator $M_{xq}$ is the identity for all labels $q$. The label $x$ is then said to be \textit{transparent}. The braiding is called \textit{nondegenerate} if the trivial label is the only transparent one.\end{defn}

\noindent Nondegeneracy can be physically motivated as follows. $R$-matrices of the form $R^{ab}$ where $a\neq b$ are not gauge-invariant, and therefore cannot correspond to measurable quantities. On the other hand, monodromies are gauge-invariant. Since the monodromy of any transparent label is trivial, there is no reason to allow for nontrivial transparent labels in our algebraic models, as they cannot be distinguished from the vacuum in practice. \\ 

\noindent However, nondegeneracy comes at a price. Let $f$ be such that $N^{ff}_{q}\in\{0,1\}$ for all $q$. $R$-matrices of the form $R^{ff}$ are gauge-invariant, and so assuming modularity has the undesirable effect of discarding theories with transparent objects $f$ such that $-1$ is an eigenvalue of $R^{ff}$ (e.g. fermions). Modular tensor categories are thus limited to describing $(2+1)$-dimensional \textit{bosonic} topological orders. Fermions are typically present in systems of interest (e.g. fractional quantum Hall liquids), and so it is desirable to have an algebraic model that is ``almost'' modular i.e. where the only nontrivial transparent object is a fermion. This has led to the development of \textit{spin} modular categories \cite{FMTCs}. 

\vspace{2mm}
\section{Concluding Remarks and Outlook}
\label{future}
The majority of this paper is devoted to considering the action of braiding on quasiparticle systems. To this end, the ``superselection braid'' proved to be central to our exposition. We saw that its action uniquely specified the superselection sectors of a system, illuminated the fusion structure amongst them and suggested the ribbon relation. Using exchange symmetry as our guiding physical principle, we showed that postulates \textbf{A1}-\textbf{A3} suffice to recover unitary ribbon fusion categories as a framework for modelling anyons. Taking into account the results of \cite{shikokim}, we also suggested an alternative set of postulates \textbf{P1}-\textbf{P3} in Section \ref{relshisec}.  \\ 

A \textit{motion group} may be defined in a more general context than that found in \mbox{Section \ref{pexchgs}} in order to describe the `motions' of a (typically disconnected) nonempty submanifold $\mathcal{N}$ in manifold $\mathcal{M}$ \cite{qiumots}. If $\mathcal{M}=\mathbb{R}^{3}$ and $\mathcal{N}$ is given by $n$ disjoint loops then the motion group is the \textit{loop braid group} $\mathcal{L}B_{n}$. Physically, we expect $\mathcal{L}B_{n}$ to play a similar role in describing the exchange statistics of loop-like excitations in $(3+1)$-dimensions to that of the braid group for point-like excitations in $(2+1)$-dimensions \cite{walkerwang}. The next possible generalisation could be to consider the statistics of knotted loops. The representation theory of motion groups and their relation to higher-dimensional TQFTs and topological phases of matter is an active area of research. In the case of loop excitations, various inroads have been made \cite{mesaros,levinloops1,levinloops2,tanti,guz1,guz2}. By formulating exchange symmetry in terms of the local representations of motion groups, the methods presented in this paper might be extended by adapting them to the setting of higher-dimensional excitations.

\newpage



\newpage
\appendix

\section{The Coloured Braid Groupoid and its Action}
\label{cbgactappx}
\begin{defn}A \textit{groupoid} with base $\mathcal{B}$ is a set $G$ with maps $\alpha,\omega:G\to\mathcal{B}$ and a partially defined binary operation $( \cdot , \cdot):G\times G \to G$ such that for all $f,g,h\in G$,
\begin{enumerate}[label=(\roman*)]
\item $gh$ is defined whenever $\alpha(g)=\omega(h)$, and in this case we have $\alpha(gh)=\alpha(h)$ and $\omega(gh)=\omega(g)$.
\item If either of $(fg)h$ or $f(gh)$ is defined then so is the other, and they are equal.
\item For each $g$, there are left and right identity elements respectively denoted by $\lambda_{g},\rho_{g}\in G$, for which we have $\lambda_{g}g=g=g\rho_{g}$. 
\item Each $g$ has an inverse $g^{-1}\in G$ satisfying $g^{-1}g=\rho_{g}$ and $gg^{-1}=\lambda_{g}$.
\end{enumerate}
Note that a group is a groupoid $G$ whose base contains a single element. 
\end{defn}

\noindent Consider the set of all possible $n$-braids where for any braid, each strand is assigned a distinct colour (and we always have the same $n$ colours to choose from). Equivalently, this may be thought of as bijectively assigning a number from $\{1,\ldots,n\}$ to each of the $n$ strands in a given braid. Thus, for any $n$-braid $b\in B_{n}$, there are now $n!$ distincly `coloured' versions of it contained in our set. \\
\noindent Under composition (i.e. stacking of braids), it is clear that our set possesses the structure of a groupoid. In this instance, the base is $\mathcal{B}=S_{\{1,\ldots,n\}}$ yielding the braid groupoid $B_{n}(\mathcal{B})$ \textit{for $n$ distinctly coloured strands}.

\begin{figure}[H]\centering \includegraphics[width=0.32\textwidth]{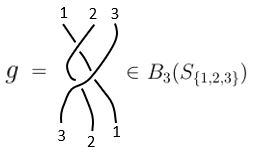}\caption{$\alpha(g)=123$ and $\omega(g)=321$}\label{basemapex}\end{figure}

\begin{remk}
Given any $s\in S_{\{1,\ldots,n\}}$, the subset of all braids $g\in B_{n}(B$) such that $\alpha(g)=\omega(g)=s$ defines a subgroup isomorphic to the pure braid group $PB_{n}$.
\label{chammaked}
\end{remk}

\noindent We can equivalently understand a groupoid $G$ with base $\mathcal{B}$ as a category $G$ whose collection of objects $\Ob(G)$ is given by $\mathcal{B}$, and where for any $x,y\in\mathcal{B}$ we have 
\begin{equation}\Hom(x,y)=\{g\in G : \alpha(g)=x , \omega(g)=y\}\end{equation}
where $g\in\Hom(x,y)$ is a morphism from $x$ to $y$. Note that all morphisms in the category $G$ are isomorphisms (by invertibility). When $\mathcal{B}=S_{\{1,\ldots,n\}}$, Remark \ref{chammaked} is equivalent to saying that there is a group isomorphism $\Aut(s)\cong PB_{n}$ for each $s\in\mathcal{B}$.\\

\noindent The categorical framework is convenient for understanding what is meant by a unitary linear representation of $B_{n}(\mathcal{B})$. In our case, this will be a functor
\begin{equation}Z:B_{n}(\mathcal{B})\to\FdHilb\label{96barsofZ}\end{equation}
where $\FdHilb$ is the category of finite Hilbert spaces, and where the image of any morphism under $Z$ is a unitary linear transformation.\footnote{$\FdHilb$ is equipped with a dagger structure given by the Hermitian adjoint.}\\

\noindent Finally, it is worth mentioning the choice of base $\mathcal{B}$ for $B_{n}(\mathcal{B})$. The coloured braid groupoid $B_{n}(\mathcal{B})$ is defined for a choice of base $\mathcal{B}=S_{\{l_{1},\ldots,l_{n}\}}$ where $l_{i}\in\mathfrak{L}$ (for some set of labels $\mathfrak{L}$). When constructing the action $\{\rho_{s}\}_{s}$ in Section \ref{exches2}, we do not assume any equalities among the representations $\{\rho_{\{i,j\}}\}_{i,j}$ in order to maintain generality. This means that all $n$ strands in any given braid must be distinctly labelled, and explains the choice of base $\mathcal{B}=S_{\{1,\ldots,n\}}$. 
\begin{itemize}
\item Suppose $\rho_{\{1,i\}}=\rho_{\{2,i\}}$ for all $i$. This is equivalent to having $\mathcal{B}=S_{\{1,1,3,4,\ldots,n\}}$ (i.e. $n-1$ colours for $n$ strands, where only $2$ strands have the same colour).
\item Suppose $\rho_{\{i,j\}}$ coincide for all $i,j$. This is equivalent to having $\mathcal{B}=S_{\{1,\ldots,1\}}$ (i.e. all strands have the same colour). In this instance, $B_{n}(\mathcal{B})\cong B_{n}$ and (\ref{96barsofZ}) is a unitary linear representation of $B_{n}$.
\end{itemize}
This suggests that a braided monoidal category $\mathcal{C}$ with $\Ob(\mathcal{C})=\mathfrak{L}$ is a sensible way to model a theory of anyons. Indeed, this is the case (anyons are algebraically modelled using braided \textit{fusion} categories). The key step is to identify the existence of a fusion structure amongst the labels in $\mathfrak{L}$ : the primary objective of this paper is to show how such structure emerges as a direct consequence of exchange symmetry.

\newpage
\section{Proofs}\label{proofsappx}
\subsection{Proofs from Section \ref{supselbraidsec}} 
\label{proofappx1}
\noindent In order to prove Lemma \ref{main1lem3}, we must first show the identities in Lemma \ref{msinmyrum}.
\begin{lem}\hspace{2mm} 
\begin{enumerate}[label=(\roman*)]
\item $\beta_{n}\sigma_{n-1}=\sigma_{1}\beta_{n} \ \ , \ n\geq2$ \label{main1lem1}
\item $b_{n}\sigma_{n-i}=\sigma_{n+1-i}\ b_{n} \ ,  \ \ \, i=1,\ldots,n-1 \ \text{\normalfont where} \ n\geq2$ \label{main1lem2}
\end{enumerate}\label{msinmyrum}\end{lem}

\begin{proof}\hspace{2mm}\\
\begin{enumerate}[label=(\roman*)]
\vspace{-3mm}
\item \phantom{For $n\geq2$,}\vspace{-7mm}
\begin{alignat*}{2} b_{n}^{2}&=b_{n-1}b_{n-2}\sigma_{n}\sigma_{n-1}\sigma_{n}&&\\
&=b_{n-1}b_{n-2}\sigma_{n-1}\sigma_{n}\sigma_{n-1}&&=b_{n-1}^{2}\sigma_{n}\sigma_{n-1}\\
&											    &&=b_{n-2}^{2}(\sigma_{n-1}\sigma_{n-2})(\sigma_{n}\sigma_{n-1})\\
&											    &&=\ldots=b_{1}^{2}\sigma_{21}\sigma_{32}\ldots(\sigma_{n}\sigma_{n-1})=\sigma_{1}b_{n}b_{n-1}\end{alignat*}
whence
\begin{alignat*}{1}\beta_{n}\sigma_{n-1}=b_{n-1}b_{n-2}\sigma_{n-1}\beta_{n-2}&=b_{n-1}^{2}\beta_{n-2}\\
&=\sigma_{1}b_{n-1}b_{n-2}\beta_{n-2}=\sigma_{1}\beta_{n}\end{alignat*}

\vspace{2.5mm}

\item For $n=2$, the identity is simply $\sigma_{121}=\sigma_{212}$. Proceeding by induction, assume that the lemma holds for some $n$. For $2\leq i\leq n$, we have
\begin{alignat*}{1}b_{n+1}\sigma_{n+1-i}=b_{n}\sigma_{n+1}\sigma_{n+1-i}&=b_{n}\sigma_{n+1-i}\sigma_{n+1} \quad \text{(where $n+1-i\in\{1,\ldots,n-1\}$)} \\
&=\sigma_{n+2-i}b_{n}\sigma_{n+1} \quad \text{(by induction hypothesis)}\\
&=\sigma_{n+2-i}b_{n+1}\end{alignat*}
For $i=1$, we show the result directly:
\begin{equation*}b_{n}\sigma_{n-1}=b_{n-2}\sigma_{n-1}\sigma_{n}\sigma_{n-1}=b_{n-2}\sigma_{n}\sigma_{n-1}\sigma_{n}=\sigma_{n}b_{n}\end{equation*}

\end{enumerate}
\end{proof}

\vspace{-4mm}
\begin{proof}[\underline{Proof of Lemma \ref{main1lem3}}]\hspace{2mm}\\[2mm]
Let us first show that
\begin{equation}\beta_{n}\sigma_{i}=\sigma_{n-i}\beta_{n}\label{nimbojimbo}\end{equation}
For $n=2$, (\ref{nimbojimbo}) is simply $\beta_{2}\sigma_{1}=\sigma_{1}^{2}=\sigma_{1}\beta_{2}$. Proceeding by induction, assume that (\ref{nimbojimbo}) holds for some $n$. For $1\leq i\leq n-1$, we have
\begin{alignat*}{1}\beta_{n+1}\sigma_{i}=b_{n}\beta_{n}\sigma_{i}&=b_{n}\sigma_{n-i}\beta_{n} \quad \text{(by induction hypothesis)}\\
														  &=\sigma_{n+1-i}b_{n}\beta_{n} \quad \text{(by Lemma \ref{msinmyrum}\ref{main1lem2})}\\
														  &=\sigma_{n+1-i}\beta_{n+1}
\end{alignat*}
For $i=n$, we want to show $\beta_{n+1}\sigma_{n}=\sigma_{1}\beta_{n+1}$, which is just Lemma \ref{msinmyrum}\ref{main1lem1}. It remains to show that
\begin{equation}\beta_{n}\sigma^{-1}_{i}=\sigma^{-1}_{n-i}\beta_{n}\label{rimbojimbo}\end{equation}
Lemma \ref{msinmyrum} implies 
\begin{subequations}
\begin{equation} \beta_{n}\sigma_{n-1}^{-1}=\sigma_{1}^{-1}\beta_{n} \ \ , \ n\geq2 \label{larkasca}\end{equation}
\begin{equation} b_{n}\sigma_{n-i}^{-1}=\sigma_{n+1-i}^{-1}b_{n} \ ,  \ \ \, i=1,\ldots,n-1 \ \text{\normalfont where} \ n\geq2 \label{larkdesca}\end{equation}
\end{subequations}
Using identities (\ref{larkasca})-(\ref{larkdesca}), the proof of (\ref{rimbojimbo}) follows similarly to that of (\ref{nimbojimbo}).
\end{proof}\vspace{7mm}

\subsection{Proofs from Section \ref{quaf}} 
\label{proofappx2}
In order to prove Theorem \ref{ssbrecur}, we will first need to prove Lemmas \ref{main2lem1}-\ref{fizko} and Proposition \ref{main2lem3}. 

\begin{lem}\begin{equation}\beta_{k}=r_{k-2}(b_{1})\cdot \ldots \cdot r_{1}(b_{k-2})\cdot r_{0}(b_{k-1}) \ , \ k\geq2\end{equation}\label{main2lem1}\end{lem}
\vspace{-5mm}
\begin{proof}
\begin{alignat*}{4}
b_{n+1}b_{n}&=\sigma_{1\ldots n+1}\cdot\sigma_{1\ldots n}\\
&=\sigma_{1\ldots n}\cdot \sigma_{1\ldots n-1}\sigma_{n+1}\sigma_{n}&&=b_{n}b_{n-1}\sigma_{n+1}\sigma_{n}\\
& &&=b_{n-1}b_{n-2}(\sigma_{n}\sigma_{n-1})(\sigma_{n+1}\sigma_{n}) \\
& &&=\ldots=b_{2}b_{1}(\sigma_{32}\cdot\ldots\cdot\sigma_{n,n-1}\cdot\sigma_{n+1,n})\\
& && =\sigma_{12}\sigma_{1}(\sigma_{32}\cdot\ldots\cdot\sigma_{n,n-1}\cdot\sigma_{n+1,n})\\
& && =\sigma_{21}(\sigma_{2}\cdot\sigma_{32}\cdot\ldots\cdot\sigma_{n,n-1}\cdot\sigma_{n+1,n})\\
& && =\sigma_{21}(\sigma_{32}\cdot\sigma_{343}\cdot\sigma_{54}\cdot\ldots\cdot\sigma_{n+1,n})\\
& && =\cdots=(\sigma_{21}\cdot\sigma_{32}\cdot\sigma_{43}\cdot\ldots\cdot\sigma_{n+1,n})\sigma_{n+1}\\
& && =\sigma_{2\ldots n+1}\cdot b_{n+1}=r_{1}(b_{n})\cdot b_{n+1}
\end{alignat*}
from which we see that
\begin{alignat*}{4}
\beta_{k}=b_{k-1}\cdot\ldots\cdot b_{1}&=(b_{k-1}b_{k-2})\cdot b_{k-3}\cdot\ldots\cdot b_{1}\\
&=r_{1}(b_{k-2})\cdot b_{k-1}\cdot b_{k-3}\cdot\ldots\cdot b_{1} \\
&=r_{1}(b_{k-2})\cdot b_{k-2}\cdot b_{k-3}\cdot\ldots\cdot b_{1}\cdot\sigma_{k-1}\\
&=r_{1}(b_{k-2})\cdot\beta_{k-1}\cdot\sigma_{k-1} \\
&=\ldots=r_{1}(b_{k-2})\cdot\ldots\cdot r_{1}(b_{1})\cdot \beta_{2} \cdot (\sigma_{2}\cdot\ldots\cdot\sigma_{k-1})=r_{1}(\beta_{k-1})\cdot b_{k-1}
\end{alignat*}
whence
\begin{alignat*}{4}
\beta_{k}&=r_{1}(\beta_{k-1})\cdot b_{k-1} \\
&=r_{1}(r_{1}(\beta_{k-2})\cdot b_{k-2})\cdot b_{k-1}&&=r_{2}(\beta_{k-2})\cdot r_{1}(b_{k-2})\cdot b_{k-1}\\
& &&=\ldots=r_{k-2}(\beta_{2})\cdot r_{k-3}(b_{2})\cdot\ldots\cdot r_{1}(b_{k-2})\cdot r_{0}(b_{k-1}) \\
\end{alignat*}
\end{proof}
\begin{lem}
\begin{equation}b_{n-1}\overleftarrow{b_{n}}=\overleftarrow{b_{n}}\cdot r_{1}(b_{n-1})\label{finny}\end{equation}
\label{aphex}\end{lem}
\vspace{-4mm}
\begin{proof}
\begin{alignat*}{4}
b_{n-1}\overleftarrow{b_{n}}&=\sigma_{1\ldots n-1}\cdot\sigma_{n\ldots 1} \\
&=b_{n-2}\cdot\sigma_{n-1}\sigma_{n}\sigma_{n-1}\cdot\overleftarrow{b_{n-2}}&&=\sigma_{n}(b_{n-2}\cdot\sigma_{n-1}\cdot\overleftarrow{b_{n-2}}) \sigma_{n} \\
& && = \sigma_{n}(b_{n-2}\cdot\overleftarrow{b_{n-1}})\sigma_{n}\\
& && = \ldots = \sigma_{n\ldots3}(b_{1}\cdot\overleftarrow{b_{2}})\sigma_{3\ldots n}\\
& && = \sigma_{n\ldots3}(\sigma_{1}\sigma_{21})\sigma_{3\ldots n} = (\sigma_{n\ldots3}\sigma_{21})(\sigma_{2}\sigma_{3\ldots n})
\end{alignat*}
\end{proof}

\begin{lem}$\beta_{n}$ is a palindrome i.e.\ $\beta_{n}=\overleftarrow{\beta_{n}}$.\label{main2lem2}\end{lem}
\begin{proof}
\begin{alignat*}{4}
\sigma_{n}\beta_{n}=\sigma_{n}b_{n-1}\beta_{n-1}&=(b_{n-2}\cdot\sigma_{n})\cdot\sigma_{n-1}\beta_{n-1}\\
&=\ldots=(b_{n-2}\cdot\sigma_{n})\cdot (b_{n-3}\cdot\sigma_{n-1})\cdot\ldots\cdot (b_{1}\sigma_{3})\cdot\sigma_{2}\beta_{2} \\
&=(b_{n-2}\cdot\ldots\cdot b_{1})(\sigma_{n}\sigma_{n-1}\cdot\ldots\cdot\sigma_{3})\sigma_{2}\sigma_{1}=\beta_{n-1}\overleftarrow{b_{n}}
\end{alignat*}
whence
\begin{alignat*}{4}
\beta_{n+1}=b_{n}\beta_{n}=b_{n-1}(\sigma_{n}\beta_{n})&=b_{n-1}\beta_{n-1}\overleftarrow{b_{n}}\\
&=b_{n-2}(\sigma_{n-1}\beta_{n-1})\overleftarrow{b_{n}}&&=b_{n-2}\beta_{n-2}\overleftarrow{b_{n-1}}\overleftarrow{b_{n}}\\
& &&=\ldots=b_{2}\beta_{2}\overleftarrow{b_{3}}\cdot\ldots\cdot\overleftarrow{b_{n}}\\
& &&= \sigma_{1}\sigma_{21}\overleftarrow{b_{3}}\cdot\ldots\cdot\overleftarrow{b_{n}}=\overleftarrow{\beta_{n+1}}
\end{alignat*}
\end{proof}
\begin{lem}\hspace{1mm}
\begin{enumerate}[label=(\roman*)]
\item $\sigma_{i}\cdot t_{k,l}=t_{k,l}\cdot r_{k}(\sigma_{i}) \ , \ \  1\leq i \leq l-1 \ , \  k\geq1 \ , \ l>1$
\item $t_{k,l}\cdot\sigma_{i}=r_{l}(\sigma_{i})\cdot t_{k,l} \ , \ \ 1\leq i \leq k-1 \ , \  k>1 \ , \ l\geq1$
\end{enumerate}\label{fizko}\end{lem}

\begin{proof} (These identities are graphically obvious; see Figure \ref{lem420} below)\hspace{1mm}
\begin{enumerate}[label=(\roman*)]

\item \textit{Claim:}\\
For $l>1$ and $j\geq0$, we have
\begin{equation}\sigma_{i}\cdot r_{j}(\overleftarrow{b_{l}})=r_{j}(\overleftarrow{b_{l}})\cdot\sigma_{i+1} \ \ , \ \ 1+j\leq i \leq (l-1)+j\label{twhite}\end{equation}
For $l=2$, (\ref{twhite}) is simply $\sigma_{1+j}(\sigma_{2+j}\sigma_{1+j})=(\sigma_{2+j}\sigma_{1+j})\sigma_{2+j}$. For $l=3$, 
\begin{equation}\begin{split}
i=1+j  \ : \ \ \sigma_{1+j}(\sigma_{3+j}\sigma_{2+j}\sigma_{1+j})&=\sigma_{3+j}(\sigma_{1+j}\sigma_{2+j}\sigma_{1+j})=(\sigma_{3+j}\sigma_{2+j}\sigma_{1+j})\sigma_{2+j} \\
i=2+j \ : \ \ \sigma_{2+j}(\sigma_{3+j}\sigma_{2+j}\sigma_{1+j})&=\sigma_{3+j}(\sigma_{2+j}\sigma_{3+j}\sigma_{1+j})=(\sigma_{3+j}\sigma_{2+j}\sigma_{1+j})\sigma_{2+j} 
\end{split}\end{equation}
Let $l\geq4$. For $2+j\leq i \leq (l-2)+j$, 
\begin{alignat*}{4}\sigma_{i}\cdot r_{j}(\overleftarrow{b_{l}})=\sigma_{i}\cdot\sigma_{l+j\ldots1+j}&=\sigma_{l+j\ldots i+2}\cdot\sigma_{i}\sigma_{i+1}\sigma_{i}\cdot\sigma_{i-1\ldots 1+j}\\
&=\sigma_{l+j\ldots i+2}\cdot\sigma_{i+1}\sigma_{i}\sigma_{i+1}\cdot\sigma_{i-1\ldots 1+j}\\
&=r_{j}(\overleftarrow{b_{l}})\cdot\sigma_{i+1}\end{alignat*}
For $i=1+j$,
\begin{equation*}\sigma_{1+j}\cdot r_{j}(\overleftarrow{b_{l}})=\sigma_{1+j}\cdot\sigma_{l+j\ldots 1+j}=\sigma_{l+j\ldots 3+j}\cdot\sigma_{1+j}\sigma_{2+j}\sigma_{1+j}=r_{j}(\overleftarrow{b_{l}})\cdot\sigma_{2+j}\end{equation*}
and for $i=(l-1)+j$,
\begin{equation*}\sigma_{(l-1)+j}\cdot r_{j}(\overleftarrow{b_{l}})=\sigma_{(l-1)+j}\sigma_{l+j}\sigma_{(l-1)+j}\cdot\sigma_{(l-2)+j\ldots 1+j}=r_{j}(\overleftarrow{b_{l}})\sigma_{l+j}\end{equation*}
This shows the claim. Recall from (\ref{tikell}) that $t_{k,l}=\left[r_{0}(\overleftarrow{b_{l}})\cdot\ldots\cdot r_{k-1}(\overleftarrow{b_{l}})\right]$. By applying the claim $k$ times for $j=0,\ldots,k-1$ (in increasing order) to $\sigma_{i}\cdot t_{k,l}$ for $1\leq i \leq l-1$, we obtain
\begin{equation}\sigma_{i}\cdot \left[r_{0}(\overleftarrow{b_{l}})\cdot\ldots\cdot r_{k-1}(\overleftarrow{b_{l}})\right]=\left[r_{0}(\overleftarrow{b_{l}})\cdot\ldots\cdot r_{k-1}(\overleftarrow{b_{l}})\right]\cdot r_{k}(\sigma_{i})\end{equation}

\vspace{2mm}
\item Applying anti-automorphism $\chi$ to (i) and relabelling yields the result.

\end{enumerate}
\end{proof}

\begin{prop}Given any positive integers $k, l$ such that $k+l\geq2$, we have\begin{enumerate}[label=(\roman*)]
\item $\beta_{k+l}=\left[r_{l}(\beta_{k})\cdot\beta_{l}\right] t_{k,l}$
\item $\beta_{k+l}=t_{l,k}\left[r_{l}(\beta_{k})\cdot\beta_{l}\right]$
\end{enumerate}
where $r_{l}(\beta_{k})$ and $\beta_{l}$ commute.\label{main2lem3}\end{prop}

\vspace{5mm}
\begin{proof}\hspace{1mm}\begin{enumerate}[label=(\roman*)]
\item By Lemma \ref{main2lem1}, we have
\begin{equation}\beta_{k+l}=r_{k+l-2}(b_{1})\cdot r_{k+l-3}(b_{2})\cdot\ldots\cdot r_{0}(b_{k+l-1})\end{equation}
and
\begin{equation}\begin{split}r_{l}(\beta_{k})&=r_{l}\left(r_{k-2}(b_{1})\cdot r_{k-3}(b_{2})\cdot\ldots\cdot r_{0}(b_{k-1})\right)\\
&=r_{k+l-2}(b_{1})\cdot r_{k+l-3}(b_{2})\cdot\ldots\cdot r_{l}(b_{k-1})\end{split}\end{equation}
whence it suffices to show that 
\begin{equation}r_{l-1}(b_{k})\cdot\ldots\cdot r_{0}(b_{k+l-1})=\left[r_{l-2}(b_{1})\cdot\ldots\cdot r_{0}(b_{l-1})\right]\cdot\left[r_{0}(\overleftarrow{b_{l}})\cdot\ldots\cdot r_{k-1}(\overleftarrow{b_{l}})\right]\label{cruxor1}\end{equation}
where the right-hand side of (\ref{cruxor1}) is $\beta_{l}\cdot t_{k,l}$. We prove (\ref{cruxor1}) by induction. \\ \\
\noindent First, perform induction on $l$ for fixed $k$. The base case $(k,l)=(k,1)$ is 
\begin{equation}r_{0}(b_{k})=r_{0}(b_{1})\cdot\ldots\cdot r_{k-1}(b_{1})\end{equation}
which is clearly true. Now suppose (\ref{cruxor1}) holds for some $l$ given fixed $k$. Then we want to show that (\ref{cruxor1}) also holds for $(k,l+1)$ i.e. 
\begin{equation}r_{l}(b_{k})\cdot\ldots\cdot r_{0}(b_{k+l})=\left[r_{l-1}(b_{1})\cdot\ldots\cdot r_{0}(b_{l})\right]\cdot\left[r_{0}(\overleftarrow{b_{l+1}})\cdot\ldots\cdot r_{k-1}(\overleftarrow{b_{l+1}})\right]\label{cruxor2}\end{equation}
Observe that 
\begin{alignat*}{2}
t_{k,l+1}&=\left[\sigma_{l+1}\cdot r_{0}(\overleftarrow{b_{l}})\right]\cdot\left[\sigma_{l+2}\cdot r_{1}(\overleftarrow{b_{l}})\right]\cdot\ldots\cdot\left[\sigma_{l+k}\cdot r_{k-1}(\overleftarrow{b_{l}})\right]\\
&=\sigma_{l+1,\ldots,l+k}\cdot\left[r_{0}(\overleftarrow{b_{l}})\cdot\ldots\cdot r_{k-1}(\overleftarrow{b_{l}})\right]=r_{l}(b_{k})\cdot t_{k,l}
\end{alignat*}
and so the right-hand side of (\ref{cruxor2}) is 
\begin{alignat*}{3}\beta_{l+1}\cdot t_{k,l+1}=b_{l}\beta_{l}\cdot r_{l}(b_{k}) t_{k,l}&=b_{l}r_{l}(b_{k})\cdot\beta_{l}t_{k,l}\\
&=b_{k+l}\cdot\beta_{l}\cdot t_{k,l}\\
&\stackrel{(\ref{cruxor1})}{=}b_{k+l}\cdot r_{l-1}(b_{k})\cdot\ldots\cdot r_{0}(b_{k+l-1})
\end{alignat*}
where the final equality follows by the induction hypothesis. Thus, in order to show (\ref{cruxor2}), we must show that
\begin{equation}r_{l}(b_{k})\cdot\ldots\cdot r_{0}(b_{k+l})=b_{k+l}\cdot r_{l-1}(b_{k})\cdot\ldots\cdot r_{0}(b_{k+l-1})\label{cruxor3}\end{equation}
under the induction hypothesis. Lemma \ref{msinmyrum}\ref{main1lem2} tells us that $b_{n}\sigma_{i}=\sigma_{i+1}b_{n}$ for any $n\geq2$ and $1\leq i \leq n-1$. Applying this result to the right-hand side of (\ref{cruxor3}), we see that $b_{k+l}$ acts on each $r_{j}$ term by $r_{1}$ as it moves to its right, yielding the left-hand side. This completes the induction on $l$. \\ 

\newpage
\noindent Next, we perform induction on $k$ for fixed $l$. The base case $(k,l)=(1,l)$ is 
\begin{equation}r_{l-1}(b_{1})\cdot\ldots\cdot r_{0}(b_{l})=\left[r_{l-2}(b_{1})\cdot\ldots\cdot r_{0}(b_{l-1})\right]\cdot r_{0}(\overleftarrow{b_{l}})\end{equation}
which we show via repeated application of Lemma \ref{aphex} on the right-hand side.
\begin{alignat*}{4}
&\left[r_{l-2}(b_{1})\cdot\ldots\cdot r_{0}(b_{l-1})\right]\cdot r_{0}(\overleftarrow{b_{l}}) \\
\stackrel{(\ref{finny})}{=}&\left[r_{l-2}(b_{1})\cdot\ldots\cdot r_{1}(b_{l-2})\cdot r_{0}(\overleftarrow{b_{l}})\right]\cdot r_{1}(b_{l-1}) \\
=&\left[r_{l-2}(b_{1})\cdot\ldots\cdot r_{2}(b_{l-3})\right]\cdot\left[r_{1}(b_{l-2})\cdot r_{1}(\overleftarrow{b_{l-1}})\sigma_{1}\right]\cdot r_{1}(b_{l-1})\\
\stackrel{(\ref{finny})}{=}&\left[r_{l-2}(b_{1})\cdot\ldots\cdot r_{2}(b_{l-3})\cdot r_{1}(\overleftarrow{b_{l-1}})\right]\cdot\left[r_{2}(b_{l-2})\sigma_{1}\right]\cdot r_{1}(b_{l-1})\\
=&\left[r_{l-2}(b_{1})\cdot\ldots\cdot r_{3}(b_{l-4})\right]\cdot\left[r_{2}(b_{l-3})\cdot r_{1}(\overleftarrow{b_{l-1}})\right]\cdot\left[r_{2}(b_{l-2})\sigma_{1}\right]\cdot r_{1}(b_{l-1})\\
=&\left[r_{l-2}(b_{1})\cdot\ldots\cdot r_{3}(b_{l-4})\right]\cdot\left[r_{2}(b_{l-3})\cdot r_{2}(\overleftarrow{b_{l-2}})\sigma_{2}\right]\cdot\left[r_{2}(b_{l-2})\sigma_{1}\right]\cdot r_{1}(b_{l-1})\\
\stackrel{(\ref{finny})}{=}&\left[r_{l-2}(b_{1})\cdot\ldots\cdot r_{3}(b_{l-4})\cdot r_{2}(\overleftarrow{b_{l-2}})\right]\cdot\left[r_{3}(b_{l-3})\cdot \sigma_{2}\right]\cdot\left[r_{2}(b_{l-2})\sigma_{1}\right]\cdot r_{1}(b_{l-1})\\
=&\ldots=r_{l-2}(b_{1})\cdot r_{l-3}(\overleftarrow{b_{3}})\cdot\left[r_{l-2}(b_{2})\sigma_{l-3}\right]\cdot\left[r_{l-3}(b_{3})\sigma_{l-4}\right]\cdot\ldots\cdot\left[r_{2}(b_{l-2})\sigma_{1}\right]\cdot r_{1}(b_{l-1})\\
=&\left[r_{l-2}(b_{1})\cdot r_{l-2}(\overleftarrow{b_{2}})\sigma_{l-2}\right]\cdot\left[r_{l-2}(b_{2})\sigma_{l-3}\right]\cdot\left[r_{l-3}(b_{3})\sigma_{l-4}\right]\cdot\ldots\cdot\left[r_{2}(b_{l-2})\sigma_{1}\right]\cdot r_{1}(b_{l-1})\\
\stackrel{(\ref{finny})}{=}&r_{l-2}(\overleftarrow{b_{2}})\cdot\left[r_{l-1}(b_{1})\sigma_{l-2}\right]\cdot\left[r_{l-2}(b_{2})\sigma_{l-3}\right]\cdot\left[r_{l-3}(b_{3})\sigma_{l-4}\right]\cdot\ldots\cdot\left[r_{2}(b_{l-2})\sigma_{1}\right]\cdot r_{1}(b_{l-1})\\
\end{alignat*}

\noindent Observe that $\sigma_{i} r_{i}(b_{l-i})=\sigma_{i}\sigma_{i+1,\ldots,l}=r_{i-1}(b_{l-i+1})$ for $1\leq i<l$, whence
\begin{alignat*}{4}
\left[r_{l-2}(b_{1})\cdot\ldots\cdot r_{0}(b_{l-1})\right]\cdot r_{0}(\overleftarrow{b_{l}})&=r_{l-2}(\overleftarrow{b_{2}})\cdot r_{l-1}(b_{1})\cdot\left[r_{l-3}(b_{3})\cdot\ldots\cdot r_{0}(b_{l})\right]\\
&=\sigma_{l,l-1,l}\cdot\left[r_{l-3}(b_{3})\cdot\ldots\cdot r_{0}(b_{l})\right]\\
&=r_{l-1}(b_{1})\cdot r_{l-2}(b_{2})\cdot\ldots\cdot r_{0}(b_{l})\\
\end{alignat*}
which proves the base case. Now suppose (\ref{cruxor1}) holds for some $k$ given fixed $l$. Then we want to show that (\ref{cruxor1}) also holds for $(k+1,l)$ i.e.

\begin{equation}r_{l-1}(b_{k+1})\cdot\ldots\cdot r_{0}(b_{k+l})=\left[r_{l-2}(b_{1})\cdot\ldots\cdot r_{0}(b_{l-1})\right]\cdot\left[r_{0}(\overleftarrow{b_{l}})\cdot\ldots\cdot r_{k}(\overleftarrow{b_{l}})\right]\label{cruzor1}\end{equation}
Observe that $t_{k+1,l}=t_{k,l}\cdot r_{k}(\overleftarrow{b_{l}})$, and so the right-hand side of (\ref{cruzor1}) is
\begin{alignat*}{4}
\beta_{l}\cdot t_{k+1,l}&=(\beta_{l}\cdot t_{k,l})\cdot r_{k}(\overleftarrow{b_{l}})\\
&\stackrel{(\ref{cruxor1})}{=}\left[r_{l-1}(b_{k})\cdot\ldots\cdot r_{0}(b_{k+l-1})\right]\cdot r_{k}(\overleftarrow{b_{l}})
\end{alignat*}
where the second equality follows by the induction hypothesis. Thus, in order to show (\ref{cruzor1}), we must show that
\begin{equation}r_{l-1}(b_{k+1})\cdot\ldots\cdot r_{0}(b_{k+l})=\left[r_{l-1}(b_{k})\cdot\ldots\cdot r_{0}(b_{k+l-1})\right]\cdot r_{k}(\overleftarrow{b_{l}})\label{excelsiboi}\end{equation}
under the induction hypothesis. For $l=1$, (\ref{excelsiboi}) is
\begin{equation}r_{0}(b_{k+1})=r_{0}(b_{k})\cdot r_{k}(\overleftarrow{b_{1}})\end{equation}
which is clearly true.

\newpage
\noindent\textit{Claim:}
\begin{equation}r_{i-1}(b_{k+l-i})\cdot r_{k+i-1}(\overleftarrow{b_{l-i+1}})=r_{k+i}(\overleftarrow{b_{l-i}})\cdot r_{i-1}(b_{k+l-i+1})\label{urghclaim}\end{equation}
where $1\leq i \leq l-1$ and $l\geq2$. Expanding the left-hand side, we get \begin{equation}\begin{split}\sigma_{i\ldots k+l-1}\cdot\sigma_{k+l\ldots k+i}&=\sigma_{i\ldots k+l}\cdot\sigma_{k+l-1\ldots k+i} \\
&=\sigma_{i\ldots k+l-2}\cdot(\sigma_{k+l-1}\cdot\sigma_{k+l}\cdot\sigma_{k+l-1})\cdot\sigma_{k+l-2\ldots k+i}\\
&=\sigma_{i\ldots k+l-2}\cdot(\sigma_{k+l}\cdot\sigma_{k+l-1}\cdot\sigma_{k+l})\cdot\sigma_{k+l-2\ldots k+i}\\
&=\sigma_{k+l}\cdot(\sigma_{i\ldots k+l}\cdot \sigma_{k+l-2\ldots k+i})
\end{split}\label{eurgh}\end{equation}
It can be shown that
\begin{equation}r_{i-1}(b_{k+l-i+1})\cdot r_{k+i-1}(\overleftarrow{b_{l-i-j+1}})=\sigma_{k+l-j+1}\left(r_{i-1}(b_{k+l-i+1})\cdot r_{k+i-1}(\overleftarrow{b_{l-i-j}})\right)\end{equation}
for $1\leq j\leq l-i$ which we can recursively apply (for $j=2$ to $l-i$) to the parenthesised expression in the last line of (\ref{eurgh}) to obtain
\begin{alignat*}{4}
\sigma_{i\ldots k+l}\cdot \sigma_{k+l-2\ldots k+i}=\sigma_{k+l-1\ldots k+i+1}\cdot\sigma_{i\ldots k+l}
\end{alignat*}
This proves the claim (\ref{urghclaim}). \\ \\ \\
We recursively apply (\ref{urghclaim}) to the right-hand side of (\ref{excelsiboi}) for \mbox{$i=1$ to $l-1$}:
\begin{alignat*}{4}
&\left[r_{l-1}(b_{k})\cdot\ldots\cdot r_{0}(b_{k+l-1})\right]\cdot r_{k}(\overleftarrow{b_{l}})\\
\stackrel{(\ref{urghclaim})}{=}&\left[r_{l-1}(b_{k})\cdot\ldots\cdot r_{1}(b_{k+l-2})\right]r_{k+1}(\overleftarrow{b_{l-1}})\cdot r_{0}(b_{k+l})\\
\stackrel{(\ref{urghclaim})}{=}&\ldots\stackrel{(\ref{urghclaim})}{=}r_{l-1}(b_{k})\cdot r_{k+l-1}(\overleftarrow{b_{1}})\cdot\left[r_{l-2}(b_{k+2})\cdot\ldots\cdot r_{0}(b_{k+l})\right]\\
&\hspace{8mm}=r_{l-1}(b_{k+1})\cdot r_{l-2}(b_{k+2})\cdot\ldots\cdot r_{0}(b_{k+l})
\end{alignat*}
which is the left-hand side\ of (\ref{excelsiboi}). This completes the induction on $k$.\vspace{8mm}

\item Applying the anti-automorphism $\chi$ to (i), we get
\begin{alignat*}{4}\overleftarrow{\beta_{k+l}}&=\overleftarrow{t_{k,l}}\left[\overleftarrow{\beta_{l}}\cdot r_{l}(\overleftarrow{\beta_{k}})\right]\\
&=t_{l,k}\left[\beta_{l}\cdot r_{l}(\beta_{k})\right]\end{alignat*}
where the second line follows by Lemma \ref{main2lem2} and $\overleftarrow{t_{k,l}}=t_{l,k}$. It is clear that $\beta_{l}$ commutes with $r_{l}(\beta_{k})$. The result follows.
\end{enumerate}\end{proof}

\begin{proof}[\underline{Proof of Theorem \ref{ssbrecur}}]\hspace{2mm}\\[2mm]
Expressions (i) and (ii) were already proved in Proposition \ref{main2lem3}. From Lemma \ref{fizko}, it easily follows that for any positive integers\footnote{Lemma \ref{fizko} implies (\ref{dow1}) and (\ref{dow2}) for $l>1$ and $k>1$ respectively. However, it is trivial to see that (\ref{dow1}) and (\ref{dow2}) also hold for $l=1$ and $k=1$ respectively.} $k,l$, we have
\begin{subequations}\begin{align}
\beta_{l}\cdot t_{k,l}&=t_{k,l}\cdot r_{k}(\beta_{l}) \label{dow1} \\
t_{k,l}\cdot\beta_{k}&=r_{l}(\beta_{k})\cdot t_{k,l} \label{dow2}
\end{align}\end{subequations}
\noindent Expressions (iii) and (iv) are implied by (i) and (ii) using either one of (\ref{dow1}),(\ref{dow2}). 
\end{proof}\vspace{5mm}

\begin{figure}[H]\centering \includegraphics[width=0.8\textwidth]{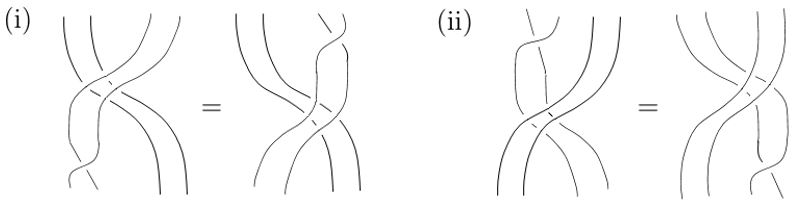}\caption{The identities in Lemma \ref{fizko} are easily seen  through braid isotopy. We illustrate these identities for $(k,l)=(2,2)$.}\label{lem420}\end{figure}

\newpage
\section{Uniqueness of the Superselection Braid}\label{ssbfbs} \vspace{3mm}

\begin{proof}[\underline{Proof of Theorem \ref{sidethm1}}]\hspace{2mm}\\[4mm]
Consider any fusion tree for an $n$-quasiparticle system. Label each of the $(n-1)$ fusion vertices in the tree with an admissible fusion outcome: in particular, the root is assigned label $Q$ corresponding to a superselection sector of the system.\\

\noindent Any superselection braid $\Lambda_{n}$ must be some composition of braids of the form $r_{d}(t_{k,l}^{\pm1})$ (since it must be compatible with the fusion trees). Recall that such braids have associated exchange phase of the form in Theorem \ref{mainthm2}(i) (and Corollary \ref{negtkl}). The statistical phase induced by $\Lambda_{n}$ should not depend on the labels of any internal vertices, and should only depend on the root label $Q$ (since the associated eigenspaces should correspond to the $n$-quasiparticle superselection sectors): we thus denote this phase by $\lambda_{n}(Q)$. We know $\Lambda_{1}=e$ and $\Lambda_{2}$ is uniquely given by $\sigma_{1}$ (up to orientation). \\

\noindent Take an arbitrary fusion vertex $v$ in the tree, and suppose that $\Lambda_{n}$ does not contain the braid that exchanges its incoming branches. This introduces the dependence of $\lambda_{n}(Q)$ on (a) the labels of the immediate children of $v$ (unless they are leaves), and (b) the labels of the parent and sibling of $v$ (unless $v$ is the root). It follows that $\Lambda_{n}$ must either (i) exchange every pair of incoming branches once, or (ii) exchange no branches. Since $\Lambda_{n}$ does not act trivially for $n>1$, it must do the former. \\  By similar considerations, we see that unless the orientation of the branch-exchanging braid acting on a fusion vertex $v$ matches that of the branch-exchanging braids acting on its parent (unless $v$ is the root) and immediate children, then $\lambda_{n}(Q)$ acquires a dependence on some labels other than $Q$. \\ \\
We thus know that $\Lambda_{n}$ must exchange every pair of incoming branches once, and that every such exchange must be oriented the same. By construction, all possible \mbox{superselection} braids have the same associated eigenspaces (namely the super-selection sectors of the system). The above further tells us that all possible super-selection braids whose orientations match have identical associated spectra $\{\lambda_{n}(Q)\}_{Q}$ (while all possible superselection braids of the opposite orientation have identical \mbox{associated} spectra $\{\lambda^{*}_{n}(Q)\}_{Q})$. \\ \\
Next, observe that any $\Lambda_{n}$ must contain the braid that exchanges the incident branches of the root node. Thus, any given $\Lambda_{n}$ of clockwise orientation must be of one or more of the following forms for any $k,l$ such that $n=k+l$ \ :
\begin{enumerate}[label=(\arabic*)]
\item $\left[\Lambda_{l}\cdot r_{l}(\Lambda_{k})\right]t_{k,l}$
\item $t_{k,l}\left[\Lambda_{k}\cdot r_{k}(\Lambda_{l})\right]$
\item $\Lambda_{l}\cdot t_{k,l}\cdot\Lambda_{k}$
\item $r_{l}(\Lambda_{k})\cdot t_{k,l}\cdot r_{k}(\Lambda_{l})$
\end{enumerate}
where for any fixed one of the above four forms, the expressions for all possible $k,l$ must be equal. By Theorem \ref{ssbrecur}, we know that all four forms are equal and are precisely $\Lambda_{n}=\beta_{n}$. \footnote{For $\Lambda_{n}$ anticlockwise, simply append a superscript `$-1$' to each $t$ in (1)-(4). By Theorem \ref{ssbrecur}, they are all equivalent to $\beta_{n}^{-1}$.}
\end{proof}

\newpage
\section{Coherence Identities}\label{penthexapx} 
\begin{enumerate}[label=(\roman*)]

\item The \textit{triangle equations} are given by\vspace{1mm}
\begin{equation}\begin{split}&\begin{gathered}\centering \includegraphics[width=1\textwidth]{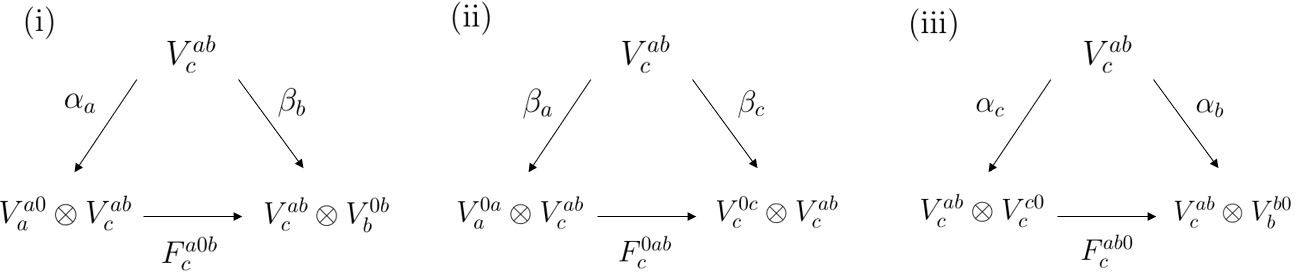}\end{gathered}   \\ 
&\hspace{5mm}\textit{commute for all $a,b,c\in\mathfrak{L}$.}\end{split}\label{treqns}\end{equation}\\
It can be shown that triangle equations (\ref{treqns}) (ii) and (iii) follow as corollaries of fundamental triangle equation (i) and the pentagon equation \cite{kitaev}. \\
Illustrating the fusion trees in (\ref{treqns}),
\begin{figure}[H]\centering \includegraphics[width=0.75\textwidth]{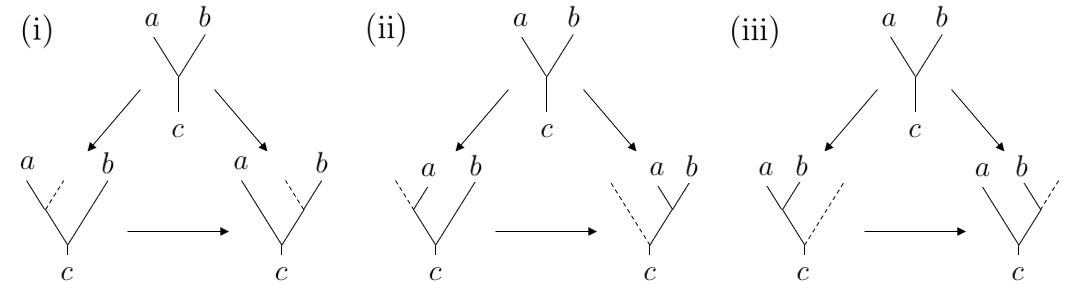} \caption*{} \label{treqns2}\end{figure}
\vspace{-12mm}
\noindent where dashed lines denote the vacuum. Independently of the gauge, symbols $F^{a0b}_{c}, F^{0ab}_{c}$ and $F^{ab0}_{c}$ correspond to the identity map\footnote{In the $6j$ fusion system formalism, this requirement is referred to as the \textit{triangle axiom} \cite{wangbook}.}. Then following (\ref{alphetiso}), it is clear that the triangle equations will be trivially satisfied.
\vspace{3mm}

\item We have the \textit{pentagon equation}\footnote{This has a nice interpretation in terms of associahedra (convex polytopes whose vertices and edges respectively correspond to distinct fusion bases and F-moves between them); see \cite{kitaev}.} :
\begin{equation}\begin{split}&\begin{gathered}\centering \includegraphics[width=1.1\textwidth]{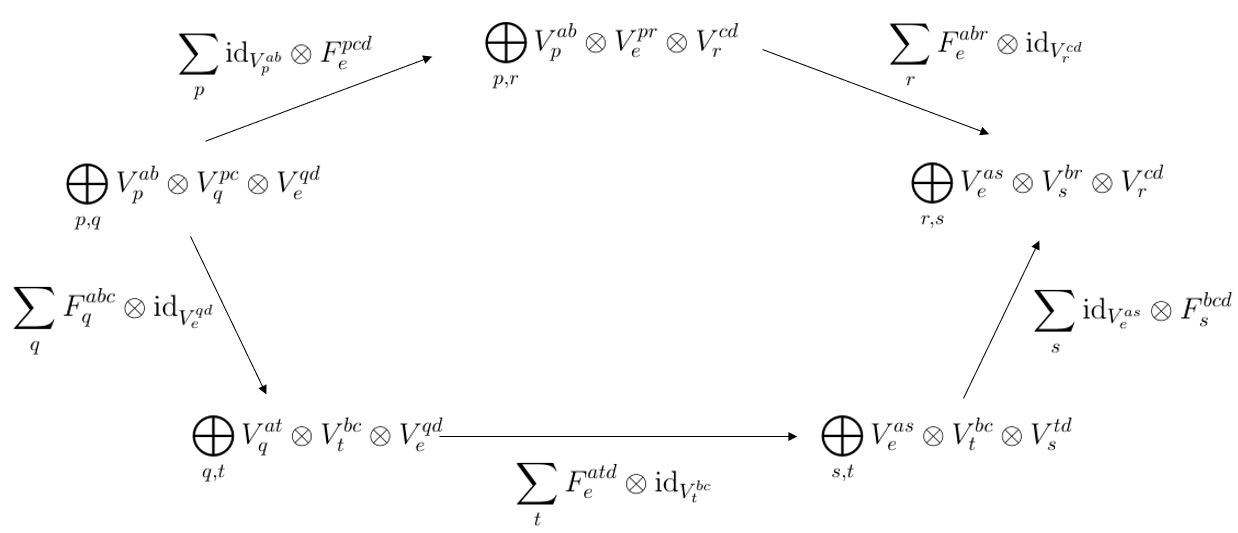}\end{gathered}   \\ 
& \hspace{12mm}\textit{commutes for all $a,b,c,d,e\in\mathfrak{L}$.}\end{split}\label{pent1}\end{equation}
\vspace{5mm}
\noindent Illustrating the fusion trees in (\ref{pent1}),\vspace{-4mm}
\begin{figure}[H]\centering \includegraphics[width=0.5\textwidth]{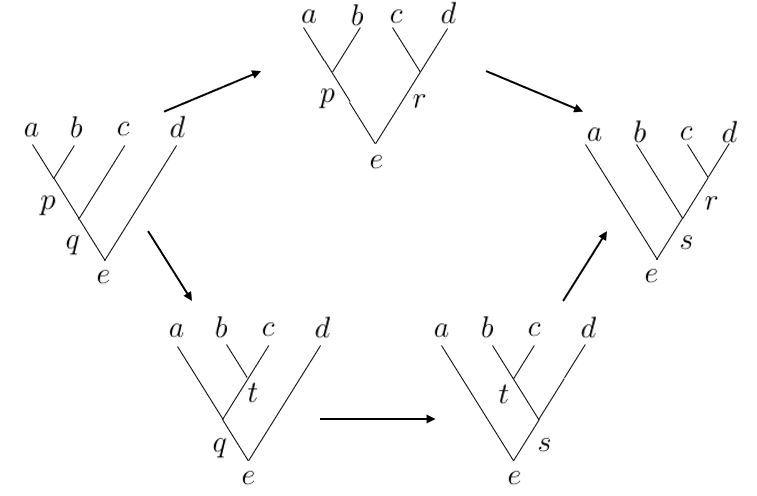} \caption*{} \label{pent2}\end{figure}
\vspace{-8mm}
\noindent The pentagon equation (\ref{pent1}) may be written
\begin{equation}\sum_{p,r}(F^{abr}_{e}\otimes\text{id}_{V^{cd}_{r}})(\text{id}_{V^{ab}_{p}}\otimes F^{pcd}_{e})=\sum_{q,s,t}(\text{id}_{V^{as}_{e}}\otimes F^{bcd}_{s})(F^{atd}_{e}\otimes\text{id}_{V^{bc}_{t}})(F^{abc}_{q}\otimes\text{id}_{V^{qd}_{e}})\label{penteq1}\end{equation}
Fixing the fusion states in the initial and terminal fusion basis, we obtain an entry-wise form of (\ref{penteq1}) which is useful for direct calculations. Fix initial state 
\[\ket{ab\to p;\alpha}\ket{pc\to q;\beta}\ket{qd\to e;\lambda}\]
and terminal state 
\[\ket{as\to e;\rho}\ket{br\to s;\delta}\ket{cd\to r;\gamma}\]
This gives us
\begin{equation}\begin{split}\sum_{\sigma}&\left[F^{abr}_{e}\right]_{(s,\delta,\rho)(p,\alpha,\sigma)}\left[F^{pcd}_{e}\right]_{(r,\gamma,\sigma)(q,\beta,\lambda)}\\
=\sum_{t,\mu,\nu,\eta}&\left[F^{bcd}_{s}\right]_{(r,\gamma,\delta)(t,\mu,\eta)}\left[F^{atd}_{e}\right]_{(s,\eta,\rho)(q,\nu,\lambda)}\left[F^{abc}_{q}\right]_{(t,\mu,\nu)(p,\alpha,\beta)}\label{penteq2}\end{split}\end{equation}
In a \textit{multiplicity-free} theory (a theory where all fusion coefficients are either $0$ or $1$), (\ref{penteq2}) is simply
\begin{equation}\left[F^{abr}_{e}\right]_{sp}\left[F^{pcd}_{e}\right]_{rq}=\sum_{t}\left[F^{bcd}_{s}\right]_{rt}\left[F^{atd}_{e}\right]_{sq}\left[F^{abc}_{q}\right]_{tp}\label{penteq3}\end{equation} 
The pentagon equation is also known as the \textit{Biedenharn-Elliot identity}.

\vspace{12mm}

\item R-matrices are transformations between bases of the form in (\ref{fusbrba}). In the graphical calculus,
\begin{equation}\begin{gathered}\centering \includegraphics[width=0.6\textwidth]{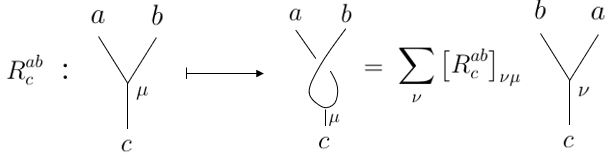}\label{rgraphginv}\end{gathered}\end{equation}
(\ref{rgraphginv}) is the gauge-free description of an R-matrix. Note that the matrix $R^{ab}$ is block-diagonal with block dimensions $\{N^{ab}_{c}\}_{c}$ . 
\newpage
We have the \textit{hexagon equations}\footnote{We roughly sketch the origin of the hexagon equations. Consider the set $F_{n}$ of $n$-leaf fusion trees. Let $\mathscr{F}_{n}$ be the set whose elements are given by those in $F_{n}$ but with all possible permutations of the string $q_{1}\ldots q_{n}$ labelling the leaves (so that $|\mathscr{F}_{n}|=n!\cdot|F_{n}|$). We define a digraph $KR_{n}$ to have vertex set $\mathscr{F}_{n}$ and edges given by all $F$ and (identically oriented) $R$ moves transforming between the elements of $\mathscr{F}_{n}$. Any pair of adjacent vertices will share precisely one edge. In order to have compatibility between all $F$ and $R$ moves, it suffices to demand that the Yang-Baxter equation is satisfied: we thus only need to consider subgraphs of the form $KR_{3}$ i.e.\ the \textit{Franklin graph}. This graph may be drawn as a dodecagon containing six hexagons and three (automatically commutative) quadrilaterals. The Yang-Baxter equation holds if the dodecagon commutes: imposing the hexagon equations ensures that the hexagons commute, and consequently that the dodecagon commutes. We remark that by restricting the edges of $KR_{n}$ to only permit $R$-moves acting on two \textit{leaves} in a direct fusion channel, we obtain the graph corresponding to the $n^{th}$ \textit{permutoassociahedron} \cite{kapranov}.} :
\vspace{-7mm}
\begin{equation}\begin{split}&\begin{gathered}\centering \includegraphics[width=1.05\textwidth]{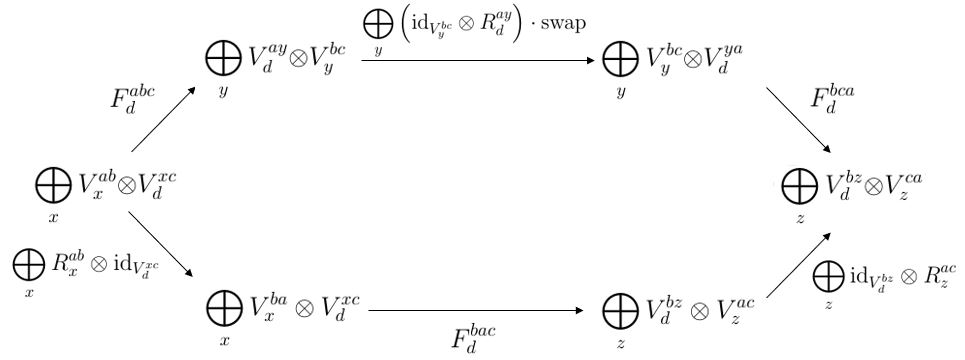}\end{gathered}   \\ 
&\begin{gathered}\centering \includegraphics[width=1.065\textwidth]{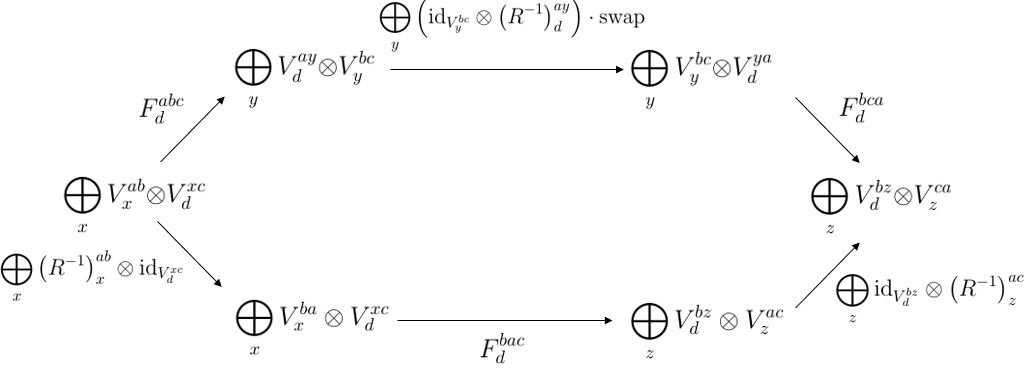}\end{gathered}   \\ 
&\hspace{10mm}\textit{commute for all $a,b,c,d\in\mathfrak{L}$.}\end{split}\label{hexeqs1}\end{equation} 

\vspace{-5mm}
\begin{figure}[H]\centering \includegraphics[width=0.45\textwidth]{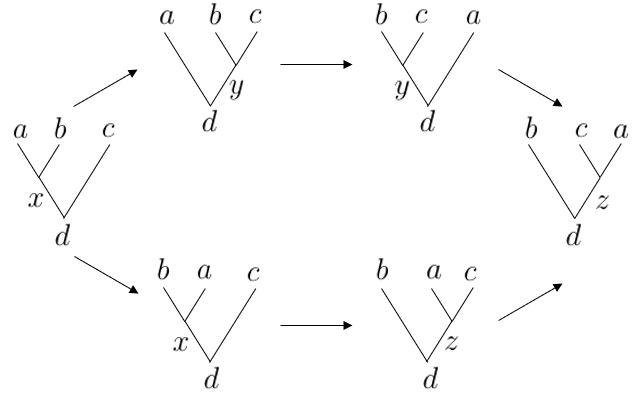} \caption{\noindent An illustration of the fusion trees in (\ref{hexeqs1}).} \label{hexeqn2}\end{figure}

\noindent Note that the only difference between the two hexagon equations is the orientation of the R-moves. Fix initial state $\ket{ab\to x;\alpha}\ket{xc\to d;\lambda}$ and terminal state $\ket{bz\to d;\rho}\ket{ca\to z;\gamma}$ in (\ref{hexeqs1}). This gives us
\vspace{2mm}
\begin{subequations}\begin{align}
\begin{split}\sum_{y,\beta,\mu,\sigma}&\left[F^{bca}_{d}\right]_{(z,\gamma,\rho)(y,\beta,\sigma)}\left[R^{ay}_{d}\right]_{\sigma\mu}\left[F^{abc}_{d}\right]_{(y,\beta,\mu)(x,\alpha,\lambda)} \\
=\sum_{\delta,\epsilon}&\left[R^{ac}_{z}\right]_{\gamma\epsilon}\left[F^{bac}_{d}\right]_{(z,\epsilon,\rho)(x,\delta,\lambda)}\left[R^{ab}_{x}\right]_{\delta\alpha}\end{split} \\
\begin{split} 
\sum_{y,\beta,\mu,\sigma}&\left[F^{bca}_{d}\right]_{(z,\gamma,\rho)(y,\beta,\sigma)}\left[\left(R^{-1}\right)^{ay}_{d}\right]_{\sigma\mu}\left[F^{abc}_{d}\right]_{(y,\beta,\mu)(x,\alpha,\lambda)} \\
=\sum_{\delta,\epsilon}&\left[\left(R^{-1}\right)^{ac}_{z}\right]_{\gamma\epsilon}\left[F^{bac}_{d}\right]_{(z,\epsilon,\rho)(x,\delta,\lambda)}\left[\left(R^{-1}\right)^{ab}_{x}\right]_{\delta\alpha}
\end{split}
\end{align}\end{subequations}
\noindent which in the construction from (\ref{rdiag1})-(\ref{rdiag2}) becomes
\begin{subequations}\begin{align}
\begin{split}
\sum_{y,\beta,\mu}&\left[F^{bca}_{d}\right]_{(z,\gamma,\rho)(y,\beta,\mu)}\left[R^{ay}_{d}\right]_{\mu\mu}\left[F^{abc}_{d}\right]_{(y,\beta,\mu)(x,\alpha,\lambda)} \\
=&\left[R^{ac}_{z}\right]_{\gamma\gamma}\left[F^{bac}_{d}\right]_{(z,\gamma,\rho)(x,\alpha,\lambda)}\left[R^{ab}_{x}\right]_{\alpha\alpha}
\end{split} \\
\begin{split}
\sum_{y,\beta,\mu}&\left[F^{bca}_{d}\right]_{(z,\gamma,\rho)(y,\beta,\mu)}\left[\left(R^{-1}\right)^{ay}_{d}\right]_{\mu\mu}\left[F^{abc}_{d}\right]_{(y,\beta,\mu)(x,\alpha,\lambda)} \\
=&\left[\left(R^{-1}\right)^{ac}_{z}\right]_{\gamma\gamma}\left[F^{bac}_{d}\right]_{(z,\gamma,\rho)(x,\alpha,\lambda)}\left[\left(R^{-1}\right)^{ab}_{x}\right]_{\alpha\alpha}
\end{split}
\end{align}\end{subequations}
and which in a multiplicity-free theory becomes
\begin{subequations}\begin{align}
\begin{split}
\sum_{y}\left[F^{bca}_{d}\right]_{zy}\left[R^{ay}_{d}\right]\left[F^{abc}_{d}\right]_{yx}=\left[R^{ac}_{z}\right]\left[F^{bac}_{d}\right]_{zx}\left[R^{ab}_{x}\right]
\end{split} \\
\begin{split}
\sum_{y}\left[F^{bca}_{d}\right]_{zy}\left[\left(R^{-1}\right)^{ay}_{d}\right]\left[F^{abc}_{d}\right]_{yx}=\left[\left(R^{-1}\right)^{ac}_{z}\right]\left[F^{bac}_{d}\right]_{zx}\left[\left(R^{-1}\right)^{ab}_{x}\right]
\end{split}
\end{align}\end{subequations} 
\end{enumerate}

\end{document}